\documentclass[10pt,a4paper]{article}
\usepackage[tbtags]{amsmath}    
\usepackage{amsfonts,amssymb,amsthm,euscript,makeidx,pifont,boxedminipage,fullpage}
\usepackage{mathtools}
\usepackage[latin1]{inputenc}
\usepackage{ucs}
\usepackage[english]{babel}
\usepackage{amsmath}
\usepackage{enumitem}

\newcommand{\ds}{\displaystyle}

\newcommand{\bc}{\begin{center}}
\newcommand{\ec}{\end{center}}
\newcommand{\1}{\mathbf{1}}
\newcommand{\E}{\mathbb{E}}

\renewcommand{\P}{\mathbb{P}}

\newcommand{\C}{\mathbb{C}}
\newcommand{\R}{\mathbb{R}}
\newcommand{\Prob}{{\mathbb P}}

\newcommand{\N}{\mathbb{N}}

\newcommand{\cF}{\mathcal{F}}

\newcommand{\cN}{\mathcal{N}}

\newcommand{\Z}{\mathbb{Z}}

\newcommand{\pare}[1]{\left ( #1 \right )}
\newcommand{\croc}[1]{\left [ #1 \right ]}

\newtheorem{theorem}{Theorem}[section]
\newtheorem{lemma}[theorem]{Lemma}
\newtheorem{proposition}[theorem]{Proposition}

\theoremstyle{definition}
\newtheorem{definition}[theorem]{Definition}
\newtheorem{remark}[theorem]{Remark}

\usepackage{color}

\reversemarginpar

\begin{document}

\title{
A probabilistic representation of the solution to a 1D evolution equation 
in a medium with negative index
}

\author{%
\'Eric Bonnetier\footnote{Universit\'e Grenoble-Alpes, Institut Fourier,  CS 40700, 38058 Grenoble Cedex 9, France. Email: 
    \textsf{eric.bonnetier@univ-grenoble-alpes.fr}}
\and
  Pierre~Etor\'e$(\star)$\footnote{Universit\'e Grenoble-Alpes, CNRS, Inria, Grenoble INP, LJK, 700 Avenue Centrale, 38401 St Martin D'Hères, France. Email:
    \textsf{pierre.etore@univ-grenoble-alpes.fr} $(\star$: corresponding author).}
  \and 
 Miguel~Martinez\footnote{Universit\'e Gustave Eiffel, LAMA, 5, boulevard Descartes
  77454 Marne-la-Vallée cedex 2,
    France. Email: \textsf{miguel.martinez@univ-eiffel.fr}}
}

\maketitle

\abstract{
In this work we investigate a 1D evolution equation involving a divergence form operator where the diffusion coefficient inside the divergence is changing sign, as in models for metamaterials.
We focus on the construction of a fundamental solution for the evolution equation,
which does not proceed as in the case of standard parabolic PDE's, since the associated
second order operator is not elliptic. 
We show that a spectral representation of the semigroup associated to the equation can be derived, 
which leads to a first expression of the fundamental solution. 
We also derive a probabilistic representation in terms of a pseudo Skew Brownian Motion (SBM).
This construction generalizes that derived from the killed SBM when the diffusion coefficient
is piecewise constant but remains positive.
We show that the pseudo SBM can be approached by a rescaled pseudo asymmetric random walk,
which allows us to derive several numerical schemes for the resolution of the PDE
and we report the associated numerical test results.
}

\vspace{1cm}

\begin{center}
{\bf Keywords:} \\

Negative Index Materials; Evolution equations; Spectral representation of semigroups; Skew Brownian motion ; Pseudo processes
\end{center}

\section{Introduction}

Over the last two decades, negative index materials (NIM)
have drawn considerable attention, due to the spectacular way in which electromagnetic,
acoustic or elastic waves may propagate in such media.
Composite materials built as mixtures of NIM's and classical dielectric material 
may indeed show resonant effects when excited at certain frequencies, in which
electromagnetic fields may concentrate near the interfaces, providing a way to
channel the fields. The amplitude of the fields may be enhanced by several orders of magnitude 
in the neighborhood of particles with negative permittivity or permeability.
The possibility of localizing and concentrating waves has  made NIM's a subject of great 
interest for many applications, for instance in communication and medical imaging.
\medskip

From the point of view of mathematical modeling, a large part of the work on NIM's has
focused on the so-called electrostatic approximation, for a composite medium made
of inclusions (or phases) of NIM's embedded into a homogeneous dielectric matrix phase.
This means that the time-harmonic Maxwell system is reduced to a diffusion for one of the 
components of either the electric or magnetic fields, provided the geometry of the device 
is assumed to have a direction of invariance, and provided the typical dimensions of the
inhomogeneities are small with respect to the incident wavelength.
In the simplest case, one remains with a transmission equation of the form
\begin{equation}\label{eq1}
\left\{\begin{array}{c}
\textrm{div}(A(x) \nabla u(x)) \;=\; f
\\
+ \;\textrm{boundary conditions}
\end{array}\right.
\end{equation}
in a bounded domain $D$ (one could also consider this equation in the whole space),
and where the conductivity $A(x)$ takes negative values in the inclusions of NIM's and a positive 
value  in the dielectric medium in which they are embedded.
\medskip

Because the conductivity changes sign in the domain, the bilinear form associated 
to the above PDE
\begin{eqnarray*}
a(u,v) &=& \int_D A(x) \nabla u(x) \nabla v(x),
\end{eqnarray*}
fails to be coercive in the natural Sobolev space $H^1(D)$ or in the appropriate
subspace $H$ that accounts for the imposed boundary conditions. 
Thus, one cannot invoke the Lax-Milgram Lemma to show existence of solutions to~(\ref{eq1}).
In many cases though, depending on the geometry of the NIM inhomogeneities,
one can show that the form $a$ is T-coercive~: there exists an invertible 
operator $T:H \longrightarrow H$ such that $(u,v) \in H \times H \longrightarrow a(u,Tv)$ is coercive.
In this case, applying the Lax-Milgram Lemma to the latter bilinear form
yields well-posedness of the PDE~\cite{BonnetChesnelCiarlet_1, BonnetChesnelCiarlet2,HMNguyen},
except possibly for some values of the negative conductivities, for which the associated
operator is not invertible and may even loose its Fredholm character.
\medskip

T-coercivity is thus a simple way of paliating the lack of coercivity. 
One of our objectives is to investigate whether T-coercivity also 
grants that the operator in~(\ref{eq1}) shares other characteristics of elliptic 
operators.
In particular, given the relation of the latter to stochastic processes,
we would like to answer the following questions~:
Does there exists a probabilistic representation of the solutions 
to~(\ref{eq1}) akin a Feynman-Kac formula ? What is the nature of
the underlying stochastic process ?
Can one design Monte Carlo type numerical schemes to approximate 
the solutions to~(\ref{eq1})~?
\medskip


In this work, we investigate these questions in a simple one-dimensional situation.
Since the probabilistic interpretation of an elliptic PDE is strongly related
to the parabolic equation whose infinitesimal generator is the associated elliptic 
operator, we consider a parabolic version of~(\ref{eq1}) of the form
\begin{eqnarray} \label{eq_evol}
\left\{ \begin{array}{lcll}
A(x)\partial_t u(t,x) &=& \ds\frac{1}{2}\partial_x\Big(A(x)\partial_x u(t,x)\Big),
&\quad x \in I, t > 0,
\\[8pt]
u(0,x) &=& u_0(x) & \quad x \in I,
\\[8pt]
u(t,\pm1) &=& 0, & t > 0,
\end{array}\right.
\end{eqnarray}
where $I$ is the interval $(-a,a)$, and where the conductivity $A$ is defined by
\begin{eqnarray*}
A(x) &=& \left\{\begin{array}{ll}
k, & x \in I^- := (-a,0) \\
1, & x \in I^+ := (0,a).
\end{array}\right.
\end{eqnarray*}
In other words, we assume that a dielectric with (positive) conductivity 1 fills in the right
part of the interval I, whereas when $k < 0$, the left part is filled with a negative
index material. 
Note that the time derivative of $u$ is multiplied by the conductivity $A(x)$, and 
that weak solutions to~\eqref{eq_evol} satisfy
 \begin{eqnarray}
\label{eq:evol-nondiv}
\left\{\begin{array}{lcll}
 \partial_tu(t,x) &=&\frac1 2 u^{\prime\prime}(t,x) ,& x\in(-a,0)\cup(0,a),\;\,t>0
 \\
 u(0,x)&=&u_0(x), & x\in(-a,a)
 \\
 u(t,a) &=& u(-t,a) \;=\; 0 &
 \end{array}\right.
\end{eqnarray}
and the transmission conditions at the interface
\begin{eqnarray}
\label{eq:CL-evol}
\left\{\begin{array}{cclcl}
u(t,0^-) &=& u(t,0^+),
\\
k \, u^\prime(t,0^-)&=& u^\prime(t,0^+).
\end{array} \right.
\end{eqnarray}

When $k > 0$, the coefficient in front of the time derivative of $u$ does
not bring significant changes to the usual parabolic setting,
and existence and uniqueness of solutions are guaranteed.
When $k < 0$ however, a `parabolic' version of~(\ref{eq1}) of the form
\begin{eqnarray*}
B(x)\partial_t u(t,x) &=& \ds\frac{1}{2}\partial_x\Big(A(x)\partial_x u(t,x)\Big),
\end{eqnarray*}
with initial datum~$u_0 \in L^2(I)$, may only have solutions with finite energy 
(i.e. in the space $L^2(0,\infty,H^1_0(I))$) if for any $x \in I$, $B(x)$ and $A(x)$ 
have the same sign~\cite{AmmariBonnetierDuca}, so here we simply choose~$B=A$.
Our objective is to construct a fundamental solution for~(\ref{eq_evol})
which can be interpreted as the transition function of a stochastic process,
or rather as we explain below, of a pseudo-process. We explicitely
characterize the associated measure, and inquire how it can be decomposed
so as to build probabilistic-like numerical schemes to compute the solutions
to~\eqref{eq:CL-evol}.
\medskip


When $k > 0$, it is well known that the fundamental solution of \eqref{eq_evol} can be interpreted as the transition function of a stochastic process : this can be shown for example by using the stochastic calculus for Dirichlet forms (see for e.g. \cite{Fukushima-et-al-2011}). 
The solution $u(t,x)$ of~\eqref{eq:evol-nondiv}-\eqref{eq:CL-evol} can be represented via a Feynman-Kac formula as the expectation 
$\E^x(u_0(\check{X}_t))$  where $\check{X}$ stands for a Skew Brownian motion (SBM) of parameter $\beta = \frac{1-k}{1+k}$, 
killed when reaching $-a$ or $a$ (we have denoted~$\E^x(\cdot)=\E(\cdot|\check{X}_0=x)$; 
see e.g. \cite{kara},  and~\cite{harrison-shepp} for an account on the Skew Brownian motion).
Thus the transition probability density function (transition function, in short) of $\check{X}$ provides a fundamental solution of~\eqref{eq_evol}.
\medskip

When $k < 0$, the skew Brownian motion may not be well defined. However, the above construction formally provides a solution to \eqref{eq:evol-nondiv}-\eqref{eq:CL-evol} that generalizes 
the probabilistic representation. 
With this perspective, we consider a {\it pseudo} expectation involving 
a {\it pseudo} Skew Brownian Motion, which we properly define as a {\it pseudo} process,
i.e., as a family of measurable functions  $Y=(Y_t)$ 
defined on a space $(\Omega,\cF)$ endowed with a signed measure $\P$ with $\P(\Omega)=1$.
The functions $Y=(Y_t)$ are called {\it pseudo}-random variables $Y=(Y_t)$,
and $(\Omega,\cF,\P)$ is referred to as a {\it pseudo}-probability space.
See~\cite{lachal2} and the references therein for an account on pseudo random variables. 
A {\it pseudo} process is defined mainly by its transition function, which allows to compute 
{\it pseudo} transition probabilities.
\medskip

When $I= \R$ ($a=+\infty$), the construction of a pseudo SBM can be successfully carried out, 
taking advantage of the fact that in the case $k>0$, 
the transition function of the SBM is known  (\cite{walsh}).
It follows that one can indeed construct a fundamental solution of~\eqref{eq_evol} 
and define a pseudo SBM in the general case $k\in\R^*\setminus\{-1\}$. 
Moreover, the (pseudo) probabilistic representation obtained in this way
naturally lends itself to numerical approximation, via the convergence of pseudo skew 
random walks to pseudo skew Brownian motions.  
\medskip

The bounded case, where $I=(-a,a), a<+\infty$, is more involved, since 
it is not clear how one could define a 'killed' pseudo skew Brownian motion 
as a pseudo-probabilistic process.
This is  mainly because 'killing' is a trajectorial operation, and because the trajectories 
of a pseudo processes do not have a clear meaning. 
We are however able to extend the representation of the solution $u(t,x)$ 
of~\eqref{eq:evol-nondiv}-\eqref{eq:CL-evol} 
in terms of the fundamental solution, when $k\in\R^*\setminus\{-1\}$. 
One preliminary step in this direction, consists in computing the transition function 
of the  killed skew Browian $\check{X}$, by probabilistic arguments in the case $k > 0$.
\medskip

This work is organized as follows. In Section~2, we construct solutions to~(\ref{eq_evol}) 
in~$L^2(0,\infty,H_0^1(I))$ for any $k \neq -1$, with the help of the 
eigenfunctions of the bilinear form $a(u,v) = \ds\int_I A(x) \partial_x u \partial_x v$
which are shown to form a basis of the space $H^1_0(I)$.
Section~3 focuses on the unbounded case when $I =\R$ (or $a = \infty$).
We show that the solutions to~(\ref{eq_evol}) can be obtained as
convolutions of the initial datum with a kernel~$\check{p}(t,x,y)$ which
has a pseudo probabilistic interpretation, and to which we associate 
a pseudo skew-Brownian motion. 
We further show that such pseudo skew-Brownian motion can be approximated 
by pseudo-skew random walks.
In Section~4, we consider the bounded case ($a < \infty$).
Finally, in Section~5, we construct several numerical schemes for solving~(\ref{eq_evol})
based on these developments and report numerical results obtained with these schemes.
\bigskip

\section{Spectral representation of the semigroup}
\label{sec:spectral}

In this section, we consider the eigenvalue problem~: find $\lambda \in \mathbb{C}$
and $u \in H^1_0(I)$ such that
\begin{eqnarray}\label{eq_eig}
- \Big(A(x) u^\prime(x) \Big) &=& \lambda^2 A(x) u(x), \quad\quad \textrm{in}\; I,
\end{eqnarray}
Seeking a solution to~(\ref{eq_eig}) in the form
\begin{eqnarray}\label{form_u}
u(x) &=& \left\{\begin{array}{lcl}
a_1 \cos(\lambda x) + b_1 \sin(\lambda x), &\quad& x \in I^-,
\\
a_2 \cos(\lambda x) + b_2 \sin(\lambda x), &\quad& x \in I^+,
\end{array}\right.
\end{eqnarray}
and expressing the transmission and boundary conditions \eqref{eq:CL-evol}, one is led
to solving the linear system
\begin{eqnarray*}
\left(\begin{array}{cccc}
1 & 0 & -1 & 0 \\
0 & k & 0 & -1 \\
\cos(\lambda a) & -\sin(\lambda a) & 0 & 0 \\
0 & 0 & \cos(\lambda a) & \sin(\lambda a)
\end{array}\right) 
\left(\begin{array}{c} a_1 \\ b_1 \\ a_2 \\ b_2 \end{array} \right)
&=& 0,
\end{eqnarray*}
from which one obtains the dispersion relation
\begin{eqnarray}\label{eq_disp}
\sin(2 \lambda a) (k+1) &=& 0.
\end{eqnarray}
Assuming $k\neq 1$, the associated eigen-elements can be grouped in two families : a set of even functions
\begin{eqnarray*}
f_{k,n}(x) &=& \cos\Big(\frac{(2n-1)\pi}{2a} x\Big), \quad x \in I,
\end{eqnarray*}
associated to $\lambda_n = \ds\frac{(2n-1)\pi}{2a}$, $n\geq 1$,
and the set of functions
\begin{eqnarray*}
g_{k,n}(x) &=& 
\left\{ \begin{array}{rl}
\sin(\frac{n \pi}{a} x) &\quad \textrm{if}\; x \in I^-,
\\[8pt]
k \sin(\frac{n \pi}{a} x) & \quad \textrm{if}\; x \in I^+,
\end{array}\right.
\end{eqnarray*}
associated to $\mu_n = \ds\frac{n \pi}{a}$, $n\geq 1$.
\medskip

When $k > 0$, the bilinear form $u,v \rightarrow \ds\int_I A(x) u(x) v(x)\,dx$ is 
a scalar product in $L^2(I)$, while the form
\begin{eqnarray*}
u,v \in H^1_0(I) &\rightarrow& a(u,v) = \ds\int_I A(x) u^\prime(x) v^\prime(x)\,dx,
\end{eqnarray*}
is coercive and symmetric in $H^1_0(I)$, and it is well-known that the solutions to~(\ref{eq_eig}) 
form a basis of $L^2(I)$ and of $H^1_0(I)$.
The next Proposition shows that this is still the case when $k < 0$.

\begin{proposition}\label{prop_basis}
When $k < 0$, the functions $(f_{k,n}, g_{k,n})_{n \geq 1}$ form a Hilbert basis
of $L^2(I)$ and of $H^1_0(I)$.
\end{proposition}

\begin{proof}
Let $k < 0$. In view of the above remark, the functions 
$(f_{-k,n},g_{-k,n})_{n \geq 1}$ are a basis of $L^2(I)$,
and one has the orthogonal decomposition 
(with respect to the scalar product associated to $-k$)
\begin{eqnarray*}
L^2(I) &=& H_f \oplus H_g,
\end{eqnarray*}
where $H_f$ (resp. $H_g$) is the vector space generated by the $f_{-k,n}$'s
(resp. by the $g_{-k,n}$'s).
Consider the mapping $T~: L^2(I) \longrightarrow L^2(I)$ defined by
\begin{eqnarray*}
Tu(x) &=& \left\{ \begin{array}{rl}
u(x) &\quad \textrm{if}\; u \in H_f,
\\
u(x) &\quad \textrm{if}\; u \in H_g \;\textrm{and}\; x \in I^-,
\\
-u(x) &\quad \textrm{if}\; u \in H_g \;\textrm{and}\; x \in I^+.
\end{array}\right.
\end{eqnarray*}
As $T \circ T = I$, $T$ is an isomorphism, thus the basis $(f_{-k,n},g_{-k,n})$
is transformed into a basis of $L^2(I)$. It is easy to check that 
$(Tf_{-k,n}, Tg_{-k,n}) = (f_{k,n}, g_{k,n})$. The same arguments show that the $(f_{k,n}, g_{k,n})$'s form a basis of $H^1_0(I)$.
\end{proof}

From now on, we assume that $k\in\R^*\setminus\{-1\}$, and we drop the index $k$ in the notation of the
basis functions $f_{k,n}$ and $g_{k,n}$. 
Note that the functions $f_n, g_n, n \geq 1$ satisfy the following relations:
\begin{eqnarray*}
\lambda_p^2
\ds\int_I A(x) f_p(x) \varphi(x) &=& \ds\int_I A(x) f_p^\prime(x) \varphi^\prime(x) \;=\; 0,
\end{eqnarray*}
for $\varphi = f_q, q \neq p$ or $\varphi = g_q, q \geq 1$, and similarly
\begin{eqnarray*}
\mu_p^2
\ds\int_I A(x) g_p(x) \varphi(x) &=& \ds\int_I A(x) g_p^\prime(x) \varphi^\prime(x) \;=\; 0,
\end{eqnarray*}
for $\varphi = g_q, q \neq p$ or $\varphi = f_q, q \geq 1$.
This corresponds to orthogonality or pseudo-orthogonality properties, depending on the sign of $k$.

It follows that any function $u_0 \in L^2(I)$ can be written in the form
\begin{eqnarray} \label{eq_u0}
u_0(x) &=& \sum_{n \geq 1} a_n f_n(x) + b_n g_n(x),
\end{eqnarray}
where the coefficients $a_n$ and $b_n$ are given by
\begin{eqnarray}
\label{eq:def-coeff-ab}
a_n \;=\; \ds\frac{ \ds\int_I A(x) u_0(x) f_n(x) \,dx}
{\ds\int_I A(x) |f_n(x)|^2 \,dx},
&\quad&
b_n \;=\; \ds\frac{ \ds\int_I A(x) u_0(x) g_n(x) \,dx}
{\ds\int_I A(x) |g_n(x)|^2 \,dx}.
\end{eqnarray}
In particular, one can check that
\begin{eqnarray}\label{eq_coeff}
\int_I A(x) |f_n(x)|^2 \, dx 
\;=\; \frac{a(k+1)}{2},
&\quad&
\int_I A(x) |g_n(x)|^2 \, dx 
\;=\; \frac{ak(k+1)}{2}.
\end{eqnarray}
\bigskip

\begin{remark}
Note that if $k=-1$ it is impossible  to provide the forthcoming representations \eqref{eq_sol} and \eqref{Green_spect}, with the $a_n$'s and $b_n$'s computed by \eqref{eq:def-coeff-ab} and \eqref{eq_coeff}. Therefore our assumption $k\in\R^*\setminus\{-1\}$.

\end{remark}

\vspace{0.3cm}

We now focus on the evolution problem~(\ref{eq_evol}) assuming that
$u_0 \in L^2(I)$. Decomposing $u_0$ 
on the basis of eigenfunctions as in~(\ref{eq_u0}), 
it is easy to check that
\begin{eqnarray} \label{eq_sol}
u(t,x) =P_tu_0(x)&=& \sum_{n \geq 1} 
a_n e^{-\lambda_n^2t/2} f_n(x) + b_n e^{- \mu_n^2 t/2} g_n(x),
\end{eqnarray}
is a weak solution to~(\ref{eq_evol}) in the sense that
$u \in L^2(0,\infty,H^1_0(I)), \partial_t u \in L^2(0,\infty, H^{-1}(I))$ and
\begin{eqnarray} \label{def_weak_sol}
\forall\; v \in H^1_0(I), \quad
\ds\int_I 2 A(x) \partial_t u(t,x)v(x) 
+ \ds\int_I A(x) \partial_x u(t,x) \partial_x v(x) &=& 0,
\end{eqnarray}
and $u(0,x) = u_0(x), a.e.\; x \in I$. With this definition, using the $f_n, g_n$'s as
test functions, the weak solution to~(\ref{eq_evol}) is easily seen to be unique.
Note that the family $(P_t)$ forms a semigroup on $L^2(I)$.
Finally, the expression~(\ref{eq_sol}) can be rewritten as
\begin{eqnarray*}
u(t,x) &=&
\int_I u_0(y) \bar{p}(t,x,y) \,dy,
\end{eqnarray*}
where the kernel has the form
\begin{eqnarray} \label{Green_spect}
\bar{p}(t,x,y)
&=& \sum_{n \geq 1} A(y)\Big( \ds\frac{2}{a(k+1)} f_n(y)f_n(x)\, e^{-\lambda_n^2 t/2}
\;+\; \ds\frac{2}{ak(k+1)} g_n(y) g_n(x)\, e^{-\mu_n^2 t/2}\Big).
\end{eqnarray}
We give a probabilistic derivation of the fundamental solution for \eqref{eq_evol} 
in  the next sections.

\section{A probabilistic construction of the fundamental solution
for the evolution equation on $\R$}
\label{sec:sol-fund-R}

In this section, we consider the case when $I = \R$, and obtain the expression 
of the fundamental solution to~\eqref{eq_evol} following a probabilistic construction.
More precisely, we consider the Cauchy problem 
\begin{equation} \label{eq_cauchy}
\left\{ \begin{array}{l}
2A(x) \partial_t u(t,x) \;=\;  (A(x) u^\prime(t,x))^\prime, \quad t > 0, \;x \in \R
\\
u(0,x) \;=\; u_0(x),
\\
u(t,\cdot) \in H^1(\R)\text{ and } Au'(t,\cdot) \in H^1(\R)\quad t > 0,
\end{array} \right.
\end{equation}
or equivalently the PDE
\begin{eqnarray*}
\partial_tu(t,x) &=&\frac1 2 u^{\prime\prime}(t,x) ,\quad x\in(-\infty,0)\cup(0,\infty),\;\,t>0,
\end{eqnarray*}
with the initial condition $u(0,x) = u_0(x), \; x\in\R$ and the 
radiation and transmission conditions
\begin{eqnarray}
\label{eq:CT-cauchy}
\left\{\begin{array}{lclclcl}
 \lim_{|x|\to\infty}u(t,x)  &=&0,&&&&
\\
u(t,0^-) &=& u(t,0^+) &\quad\quad& k u^\prime(t,0^-)&=& u^\prime(t,0^+).
\end{array} \right.
\end{eqnarray}
We first introduce a few notations. 
\vspace{0.4cm}

{\bf Notations.} We denote by $g(t,x,y)=\frac{1}{\sqrt{2\pi t}}\exp\Big(-\frac{|y-x|^2}{2t} \Big)$ the density of a $\cN(x,t)$. 

Let $b < c \in \R$. The transition density of the Brownian motion killed at the points
$b$ or $c$ is given by
\begin{equation}
\label{eq:Bro-killed}
p^{(b,c)}_W(t,x,y):=\sum_{n=-\infty}^\infty\big[g(t,x,y-2n(c-b))-g(t,x,2b-y-2n(c-b))\big]
\end{equation}
(see \cite{borodin}, Appendix I, Nï¿½ 6). Note that one has for any $y\in(b,c)$,
\begin{equation}
\label{eq:Bro-killed-lim}
\lim_{x\xrightarrow[>]{}b}p^{(b,c)}_W(t,x,y)=0\quad\text{  and  }\quad\lim_{x\xrightarrow[<]{}c}p^{(b,c)}_W(t,x,y)=0
\end{equation}

\subsection{The fundamental solution when $k>0$}

We first assume that $k>0$, so that equation~\eqref{eq_evol} is parabolic.
It is known that in this case, under mild assumptions on the initial condition $u_0$
(for example $u_0$ is continuous and bounded, see \cite{lejay-2006})
the solution to the above Cauchy problem is given by
\begin{equation}
\label{eq:FK1}
u(t,x)=\E^x[u_0(X_t)],
\end{equation}
where $X$ is the SBM with parameter $\beta=\frac{1-k}{1+k}\in(-1,1)$. 
We refer to the survey~\cite{lejay-2006} for the definition and main properties of the SBM.
In particular $X$ solves the Stochastic Differential Equation (SDE) with local time
\begin{equation}
\label{eq:SBM}
dX_t=dW_t+\beta dL^0_t(X),
\end{equation}
where $W$ denotes a standard Brownian motion driving the SDE, 
and where $L^0_t(X)$ is the symmetric local time at the point zero and at time $t$ of $X$.
Note that $|\beta|<1$ ensures the existence of $X$, see e.g. \cite{lejay-2006}.
The SBM behaves like a Brownian motion, except at the times when it touches zero, 
at which its dynamics are biased by the term of local time in \eqref{eq:SBM}. 
In particular we have:
 
\begin{lemma}
\label{lem:walsh}[Walsh, \cite{walsh}]
Let $\beta\in (-1,1)$ and let $X$ be the solution to \eqref{eq:SBM}. 
Under $\P^0$ one has:

i) The process $|X|$ is distributed as a reflecting Brownian motion $|W|$ (starting from zero).

ii) The processes  $(\mathrm{sign}(X_t))$ and $(|X_t|)$ are independent.

In addition, for any $t>0$ one has $\P^0(X_t>0)=\frac{1+\beta}{2}$.
\end{lemma}

From this Lemma and the reflection principle for the Brownian motion,
Walsh was able to give an explicit expression of the transition probability 
density $p(t,x,y)$ of the SBM, in the form~\cite{walsh}
\begin{equation}
\label{eq:walsh1}
p(t,x,y)=\left\{
\begin{array}{lll}
(1-\beta)g(t,x,y)&\text{if}&x\geq 0,\,y<0\\
\\
g(t,x,y)+\beta g(t,x,-y)&\text{if}&x>0,\,y>0\\
\\
(1+\beta)g(t,x,y)&\text{if}&x\leq 0,\,y>0\\
\\
 g(t,x,y)-\beta g(t,x,-y)&\text{if}&x<0,\,y<0.\\
\end{array}
\right.
\end{equation}
Note that \eqref{eq:walsh1} can also be written as
\begin{equation}
\label{eq:walsh2}
p(t,x,y)=\left\{
\begin{array}{lll}
(1-\beta)g(t,x,y)&\text{if}&x\geq 0,\,y<0\\
\\
-\beta [g(t,x,y)-g(t,x,-y)]+(1+\beta)g(t,x,y)&\text{if}&x>0,\,y>0\\
\\
(1+\beta)g(t,x,y)&\text{if}&x\leq 0,\,y>0\\
\\
 \beta[g(t,x,y)-g(t,x,-y)]+(1-\beta)g(t,x,y)&\text{if}&x<0,\,y<0.\\
\end{array}
\right.
\end{equation}
where $g(t,x,y)-g(t,x,-y)=:\check{g}(t,x,y)$, $x,y>0$ (resp. $x,y<0$) can
be interpreted as the transition function of a Brownian motion on 
$(0,\infty)$ (resp. $(-\infty,0)$) killed at zero (see \cite{borodin} 
and Remark \ref{rem:coherent-walsh}).

Yet another way to rewrite \eqref{eq:walsh1} is to consider 
$\hat{g}(t,x,y):=g(t,x,y)+g(t,x,-y)$, $x,y>0$, the transition function 
of a reflected Brownian motion on $[0,\infty)$ (see \cite{borodin}), which yields
\begin{equation}
\label{eq:walsh3}
p(t,x,y)=\left\{
\begin{array}{lll}
(1-\alpha)[\hat{g}(t,x,-y)-\check{g}(t,x,-y)]&\text{if}&x\geq 0,\,y<0\\
\\
\alpha\hat{g}(t,x,y)+(1-\alpha)\check{g}(t,x,y)&\text{if}&x>0,\,y>0\\
\\
\alpha[\hat{g}(t,-x,y)-\check{g}(t,-x,y)] &\text{if}&x\leq 0,\,y>0\\
\\
 (1-\alpha)\hat{g}(t,-x,-y)+\alpha\check{g}(t,-x,-y)&\text{if}&x<0,\,y<0,\\
\end{array}
\right.
\end{equation}
where we have set
\begin{eqnarray}
\label{eq:def-alpha}
\alpha &=& (1+\beta)/2=1/(1+k).
\end{eqnarray}

It follows from~\eqref{eq:FK1} that
\begin{eqnarray} \label{rep_Cauchy}
u(t,x)=\E^x[u_0(X_t)]=\int_\R u_0(y)p(t,x,y)dy,
\end{eqnarray}
so that $p(t,x,y)$ identifies with the fundamental solution of \eqref{eq_cauchy}. 

This fact also follows by mere computation: indeed
one can easily check that the function $p$, defined by~\eqref{eq:walsh1}, satisfies for any $y\in\R$
\begin{equation}
\label{eq:fund-lap-SBM}
\partial_tp(t,x,y)=\frac12\partial^2_{xx}p(t,x,y),\quad \forall x\in(-\infty,0)\cup(0,\infty).
\end{equation}
Using  $\partial_{x}g(t,x,y)=\frac{y-x}{t}g(t,x,y)$ one can check the transmission conditions
\begin{equation}
\label{eq:fund-trans-SBM}
p(t,0-,y)=p(t,0+,y),\quad\quad k\partial_{x}p(t,0-,y)=\partial_{x}p(t,0+,y),
\end{equation}
and the radiation condition is also easy to check:
\begin{equation}
\label{eq:fund-lim-SBM}
\lim_{|x|\to\infty}p(t,x,y)=0.
\end{equation}
Taking derivatives of the integral on the right-hand side of \eqref{rep_Cauchy}
and using~\eqref{eq:fund-lap-SBM}-\eqref{eq:fund-lim-SBM} shows that $p$
is indeed the fundamental solution of \eqref{eq_cauchy}. 
\medskip

\subsection{The fundamental solution in the case $k<0$ and the pseudo SBM}
\label{ssec:pseudo-SBM}

When $k\in\R_-^*\setminus\{-1\}$, setting $\beta=\frac{1-k}{1+k}$,
we may again define a function $p$ as in~\eqref{eq:walsh1}. 
It is easy to check that this function solves 
\eqref{eq:fund-lap-SBM}-\eqref{eq:fund-lim-SBM}, so that
\begin{eqnarray*}
u(t,x) &=& \ds\int_\R u_0(y) p(t,x,y) \,dy,
\end{eqnarray*}
is a solution to the Cauchy problem~\eqref{eq_cauchy}.
Consequently, $p(t,x,y)$ is a fundamental solution to \eqref{eq_cauchy} 
also in the case $k\in \R_-^* \setminus \{-1\}$.
Note that $\beta$~is not in $(-1,1)$, so that $p(t,x,y)\,dy$ is only a 
signed measure, and~\eqref{eq:SBM} does not define a SBM (see e.g.~\cite{legall}). 
Stochastic Differential equations of the type 
\begin{equation}
\label{eq:port}
dX_t=\sigma(X_t)dW_t+\beta dL^0_t(X)
\end{equation}
 with $\beta\notin (-1,1)$ have been addressed for example in
\cite{kopytko}. But this latter work does not allow to take~$\sigma\equiv 1$, $\beta\notin (-1,1)$ and to get a weak solution to \eqref{eq:port} (indeed to touch $\beta\notin (-1,1)$ the coefficient~$\sigma$ has to be different from $1$).
However, we can still define a function~$p(t,x,y)$ by \eqref{eq:walsh1} or~\eqref{eq:walsh2},
to which we can associate a pseudo SBM, as we describe below.
\medskip

By {\it pseudo-random variable}, we mean a measurable function defined 
on a space (so-called pseudo-probability space) endowed
with a signed measure with a total mass equal to 1. 
We observe that in the case $k\in\R_-^*\setminus\{-1\}$, the signed asymmetric heat-type kernel $p(t,x,y)$ defined in~\eqref{eq:walsh2} is no longer positive everywhere, 
however it integrates to 1 w.r.t. $dy$. 
In view of $\eqref{eq:walsh2}$, this is clearly the case when $x=0$. 
When $x>0$, we have
\begin{align*}
\int_{-\infty}^{+\infty} p(t,x,y)dy &= -\beta\int_{0}^{+\infty}[g(t,x,y)-g(t,x,-y)]dy+(1+\beta)\int_{0}^{+\infty}g(t,x,y)dy + (1-\beta)\int_{-\infty}^{0}g(t,x,y)dy\\
&=\beta\int_{0}^{+\infty}g(t,x,-y)dy + \int_{0}^{+\infty}g(t,x,y)dy + (1-\beta)\int_{-\infty}^{0}g(t,x,y)dy\\
&=\int_{-\infty}^{+\infty}g(t,x,y)dy + \beta\int_{0}^{+\infty}g(t,x,-y)dy -\beta\int_{-\infty}^{0}g(t,x,y)dy\\
&= 1 + \beta\pare{-\int_{0}^{-\infty}g(t,x,z)dz -\int_{-\infty}^{0}g(t,x,y)dy} = 1.
\end{align*} 
A similar computation applies to the case $x<0$.
One may also check that the Chapman-Kolmogorov idendity for the family of transition kernels 
$\pare{p(t,x,y)dy}_{x\in \R}$ is also preserved in this case $k<0, k \neq -1$.

Hence, in accordance to the usual Markov rules, we may define the pseudo-skew Brownian motion as the pseudo-Markov process `$(X_t)_{t\geq 0}$' associated with the signed asymmetric heat-type kernel $p(t,x,y)dy$ defined in \eqref{eq:walsh2} -- which is the fundamental solution 
to~\eqref{eq_cauchy} also in the case $k\in\R_-^*\setminus\{-1\}$ -- by~:~
for $t>0$ and $0=t_0<t_1<\dots t_m$ and $x=x_0, x_1,\dots, x_m, y\in \R$,
\[
\P^{x}\pare{X_t\in dy} = p(t,x,y)dy
\]
and
\[
\displaystyle\P^{x}\pare{X_{t_1}\in dx_1, \dots, X_{t_m}\in dx_m} = \prod_{i=1}^m  p(t_{i}-t_{i-1},x_{i-1},x_i)dx_i.
\]
Note that since pseudo-Markov processes are defined in terms of a signed measure, 
it is not clear how one could generalize the definition of the skew Brownian motion 
over all $t \geq 0$ in this context. 
In particular, from a strict probabilistic point of view,  the notion of trajectory for pseudo Markov processes indexed by continuous time does not  have a clear meaning.

For more results on pseudo-Markov processes, we refer to e.g.  \cite{lachal2}.
\vspace{0,3 cm}

We now give a definition, inspired from~\cite{lachal2},
for the convergence of a family of pseudo-processes $((Y_t^{\varepsilon})_{t\geq 0})_{\varepsilon>0}$ towards a pseudo-process $(Y_t)_{t\geq 0}$.

\begin{definition}\label{def_pseudo_weak_cv}
Let $((Y_t^{\varepsilon})_{t\geq 0})_{\varepsilon>0}$ denote a family of pseudo-processes and $(Y_t)_{t\geq 0}$ a pseudo-process.

We write  
\[
(Y_t^{\varepsilon})_{t\geq 0}\xrightarrow[\varepsilon \searrow 0]{\text{pseudo-w}}(Y_t)_{t\geq 0}
\]
if 
\begin{equation}
\label{eq:def-pseudo-conv}
\forall \ell\in \N^{\ast}, \forall t_1,\dots, t_\ell\geq 0, \forall u_1,\dots, u_\ell\in \R,\;\;\E\croc{\exp\pare{i\sum_{j=1}^\ell u_jY_{t_j}^{\varepsilon}}}\xrightarrow[\varepsilon \searrow 0]{}\E\croc{\exp\pare{i\sum_{j=1}^\ell u_jY_{t_j}}}.
\end{equation}
\end{definition}
\begin{remark}
Note that in \eqref{eq:def-pseudo-conv}, the left hand side expectation symbol might depend on $\varepsilon$. But in order to avoid cumbersome notations, and as it will cause no ambiguity in the proofs, we simply denote it by $\E$.
\end{remark}
\vspace{0,3 cm}

\subsection{Convergence of the scaled pseudo asymmetric random walk to the pseudo SBM}
\label{sssec:pseudo-RW}

In this section, we assume that $k\in\R^*\setminus\{-1\}$, and define $\alpha$ 
as in~\eqref{eq:def-alpha}. If $k>0$ then $\alpha\in(0,1)$ and we deal with true processes and random walks. If $k\in\R^*_-\setminus\{-1\}$ then $\alpha\in\R\setminus (0,1)$ and we deal with pseudo processes and pseudo random walks. Our computations englobe both cases, but the results are new for the case $\alpha \in\R\setminus (0,1)$ (for the case 
$\alpha\in(0,1)$ see e.g. \cite{harrison-shepp}).

Before getting deeper into our subject, let us give here a brief account of the several notations that we are going to use in the remaining of this section.

\begin{itemize}[label = $*$]
\item Regarding processes :
\begin{itemize}[label = $-$]
\item $W$ stands for a standard Brownian motion (classical symmetric).
\item $M$ stands for a standard symmetric random walk on $\Z$.
\item $S$ stands for the pseudo random walk on $\Z$ that is $\alpha$-skewed at $0$.
\item $X^n$ stands for the properly normalized pseudo random walk $\alpha$-skewed at $0$.
\item $X$ denotes the pseudo-process under study. 
\item The superscript $\dagger$ denotes killing at zero: this path operation will be only performed for classical processes (either the symmetric random walk $M$ or the Brownian motion $W$).
\end{itemize} 
\end{itemize}
For example, it should be clear that $W^\dagger$ denotes a standard Brownian motion killed when it hits zero for the first time.
\begin{itemize}[label = $*$]
\item Regarding functions : 
\begin{itemize}[label = $-$]
\item The subscript $_{0}$ (or superscript $^{0}$) will always refer to a quantity concerning the classical standard symmetric $\Z$-valued random walk $M$.
\item The hat superscript ' ${\hat{ }}$ ' will always refer to quantities that are related to a reflected process at $0$ (the standard symmetric random walk on $\Z$ or the standard Brownian motion).
\item The check superscript ' ${\check{}}$ ' will always refer to quantities that are related to a process with killing at $0$ (the standard symmetric random walk on $\Z$ or the standard Brownian motion). 
\item The letters $\psi$, $\Psi$ denote characteristic functions.
\item The order of variables in a characteristic function will always be : time / starting point / argument.
\end{itemize}
\end{itemize}
For example, it should be plain that $\check{\psi}_{0}(j;m_0,v)$ refers to the characteristic function of the classical standard symmetric random walk killed at zero and starting from $m_0$, taken at time $j$ and evaluated at $v$, namely $\E^{m_0}[{\exp}(iv\,M^\dagger_j)]$. 
\medskip

These notations will be recalled when needed below.
\medskip
\medskip

Let $m_0\in \Z$ denote an arbitrary integer. Constructed on some pseudo probability space $(\Omega, {\cal F}, {\Prob}^{m_0})$ we consider~$\pare{S_n}_{n\geq 0}$ the pseudo skewed random walk on the integers starting ${\Prob}^{m_0}$-a.s. from $S_0 = m_0$ with pseudo transition probabilities given by
\begin{eqnarray}
\label{eq:pseudo-transitions-srw}
&\Prob^{m_0}\pare{S_{n+1} = S_{n} + 1 | S_0,\dots, S_n} &= \left \{\begin{array}{ll}
\ds \alpha \;\;&\;\;\;\text{if }\;S_n= 0\\
\ds {1}/{2}\;\;&\;\;\;\text{otherwise}.
\end{array}
\right .\\
&\Prob^{m_0}\pare{S_{n+1} = S_{n} - 1 | S_0,\dots, S_n} &= \left \{\begin{array}{ll}
\ds 1-\alpha\;\;&\text{if }\;S_n= 0\\
\ds {1}/{2}\;\;&\text{otherwise}.
\end{array}
\right .
\end{eqnarray}
We attach to $\pare{S_n}_{n\geq 0}$ its natural filtration $\pare{{\cal F}_n^{S}}_{n\geq 0}$ defined by $\mathcal{F}^S_n := \sigma\pare{S_k~:~k\leq n}$ for $n\geq 0$. The pseudo random sequence $\pare{S_n}_{n\geq 0}$ constructed likewise is a pseudo markovian random walk. 

\begin{remark}
\label{Remarque-importante}
Note that -- contrary to the case of pseudo processes in continuous time -- any ${\cal F}_n^S$ adapted functional of the pseudo skewed random walk $\pare{S_k}_{k\geq 0}$ can be expressed as a functional of
$(S_0,S_1,\dots, S_n)$, hence a finite dimensional pseudo distribution functional. Consequently, there is no ambigu\"ity in what is meant by the pseudo law under $\Prob^{m_0}$ of a pseudo skewed random walk on the whole path set $\Z^\N$.
\end{remark}

The following reflection principle is the key to get to the main result of this section.
\begin{lemma}
\label{prop:alphaSRW}
 For all $j\in \N^\ast$ and $m\in \Z^{\ast}$,
 \begin{align}
 \label{eq:alphaSRW-0}
 &\Prob^{0}\pare{S_{j} = m} = \left \{\begin{array}{ll}
 \ds \alpha\Prob^0\pare{|S_{j}| = m}\;\;&\text{if }\;m>0\\
 \\
\ds \pare{1-\alpha}\Prob^{0}\pare{|S_{j}| = |m|}\;\;&\text{if }\;m<0.
\end{array}
\right .
 \end{align}
Moreover, for any integer $m_0>0$
  \begin{align*}
 &\Prob^{m_0}\pare{S_{j} = m} = \left \{\begin{array}{ll}
 \ds \alpha\Prob^{m_0}\pare{|S_{j}| = m} + \pare{1-\alpha}\Prob^{m_0}\pare{|S_{j}| = m \,;\,\forall n\in [\![1,j]\!],\;|S_n|\neq 0}\;\;&\text{if }\;m>0\\
 \\
\ds \pare{1-\alpha}\Big (\Prob^{m_0}\pare{|S_{j}| = |m|} - \Prob^{m_0}\pare{|S_{j}| = m \,;\,\forall n\in [\![1,j]\!],\; |S_n|\neq 0}\Big )\;\;&\text{if }\;m<0
\end{array}
\right .
\end{align*}
 and for any integer $m_0<0$,
 \begin{align*}
 &\Prob^{m_0}\pare{S_{j} = m} = \left \{\begin{array}{ll}
 \ds (1-\alpha)\Prob^{|m_0|}\pare{|S_{j}| = |m|} + \alpha\Prob^{|m_0|}\pare{|S_{j}| = |m| \,;\,\forall n\in [\![1,j]\!],\;|S_n|\neq 0}\!\!\!&\text{if }\;m<0\\
 \\
\ds \alpha\pare{\Prob^{|m_0|}\pare{|S_{j}| = m} - \Prob^{|m_0|}\pare{|S_{j}| = |m| \,;\,\forall n\in [\![1,j]\!],\; |S_n|\neq 0}}\!\!\!&\text{if }\;m>0.
\end{array}
\right .
 \end{align*}
 \end{lemma}
\begin{proof}
We only treat the case where the pseudo random walk starts from $0$.
The other cases, although tedious, can be analysed in the same fashion.

Let $n\in \N^\ast$ and $m\in \Z^{\ast}$ fixed. 

Let us introduce $G_n= \sup\pare{j\leq n~:~S_j=0} = \sup\pare{j\leq n~:~|S_j|=0}$ that is an ${\cal F}_n^S$-measurable functional on the path space. It is easy to check that the pseudo-path $\pare{S_{G_n+j}}_{j\in \{0,\dots, n-{G}_n\}}$ has the same pseudo-law under the pseudo conditional probability $\Prob^{0}\pare{.|S_{G_n+1} = 1}$ as the pseudo-path $\pare{|S_{{G}_n+j}|}_{j\in \{0,\dots, n-G_n\}}$ under the original pseudo probability $\Prob^{0}$ (this assertion makes sense remembering Remark \ref{Remarque-importante}). Observe also that $\Prob^{0}\pare{S_{{G}_{n}+1}=1}=\alpha$.

If $m>0$, then a.s. on the set $\{S_{n} = m\}$, the value of $G_n$ cannot be equal to $n$, so that $G_n\leq n-1$ a.s. conditionally on this set. Hence, for $m>0$ we have 
\begin{align*}
\Prob^{0}\pare{S_{n} = m} &= \Prob^{0}\pare{\{S_{n} = m\}\cap \{S_{G_n+1}=1\}}\\
&=\Prob^{0}\pare{S_{G_n + (n-G_n)} = m~|~S_{G_n+1}=1}\Prob\pare{S_{G_n+1}=1}\\
&=\Prob^{0}\pare{|S_{{G}_n + (n-{G}_n)}| = m}\Prob^{0}\pare{S_{G_n+1}=1}\\
&=\alpha\Prob^{0}\pare{|S_{n}| = m}.
\end{align*}
The case $m<0$ is proved in a similar way. We have $\Prob^{0}\pare{S_{{G}_{n}+1}=1}=k/(1+k)$ and it is easily checked that the pseudo-path $\pare{S_{G_n+j}}_{j\in \{0,\dots, n-G_n\}}$ has the same pseudo-law under the conditionnal probability $\Prob^{0}\pare{.|S_{G_n+1} = -1}$ as the pseudo-path $\pare{-|S_{{G}_n+j}|}_{j\in \{0,\dots, n-{G}_n\}}$ under $\Prob^{0}$ (again, this assertion makes sense because of Remark \ref{Remarque-importante}).
\end{proof}

For an arbitrary $m_0\in \Z$, we denote by $(M_k)_{k\geq 0}$ a standard symmetric random walk on $\Z$ constructed on $(\Omega, {\cal F}, {\Prob}^{m_0})$ such that $M_0=m_0$ ${\Prob}^{m_0}$-a.s. 

Let $b{\cal E}\pare{\Z,\C}$ denote the set of bounded complex valued functions defined on $\Z$. We introduce the following family of operators $\pare{T_j^0}_{j\in \N}$ acting from $b{\cal E}\pare{\Z,\C}$ to $b{\cal E}\pare{\Z,\C}$ and defined by
\begin{align}
T_j^{0}f~:~m\mapsto 
\begin{cases}
\alpha \E^{m}f(|M_j|) + (1-\alpha)\E^{m}f(-|M_j|) + (1-\alpha)\pare{\E^{m}f(|M^\dagger_j|) - \E^{m}f(-|M^\dagger_j|)}\text{if }m\geq 0\\
\alpha \E^{|m|}f(|M_j|) + (1-\alpha)\E^{|m|}f(-|M_j|) + \alpha\pare{\E^{|m|}f(-|M^\dagger_j|) - \E^{|m|}f(|M^\dagger_j|)}\;\;\text{if }m< 0.
\end{cases}
\end{align}
We have the following lemma.
\begin{lemma}
\label{lemma:semi-group-srw}
For any function $f\in b{\cal E}\pare{\Z,\C}$, 
\[
\E^{m}(f(S_j)) = T_j^{0}f(m).
\]
\end{lemma}
\begin{proof}
Let $m_0\in \N$ arbitrary. It is plain from the definition of the transitions of the pseudo skewed random walk $\pare{S_n}_{n\geq 0}$ (given in \eqref{eq:pseudo-transitions-srw}) that $\pare{|S_n|}_{n\geq 0}$ and $\pare{|M_n|}_{n\geq 0}$ share the same pseudo law on the path space $\N^\N$ under $\P^{m_0}$: this assertion makes sense recalling Remark \ref{Remarque-importante}. In particular, we are allowed to rewrite the equalities of Lemma \ref{prop:alphaSRW} in the following manner
  \begin{align}
  \label{eq:loi-classique-remplace-pseudo}
 &\Prob^{m_0}\pare{S_{j} = m} = \left \{\begin{array}{ll}
 \ds \alpha\Prob^{m_0}\pare{|M_{j}| = m} + \pare{1-\alpha}\Prob^{m_0}\pare{|M_{j}| = m \,;\,\forall n\in [\![1,j]\!],\;|M_n|\neq 0}&\!\!\!\!\text{ if }m>0\\
 \\
\ds \pare{1-\alpha}\pare{\Prob^{m_0}\pare{|M_{j}| = |m|} - \Prob^{m_0}\pare{|M_{j}| = |m| \,;\,\forall n\in [\![1,j]\!],\; |M_n|\neq 0}}&\!\!\!\!\!\text{ if }\;m<0.
\end{array}
\right .
\end{align}
For the same reasons, we also have $\Prob^{m_0}\pare{S_{j} = 0} = \Prob^{m_0}\pare{|S_{j}| = 0} = \Prob^{m_0}\pare{|M_{j}| = 0}$.

The announced equality follows then easily by taking the expectation for $f\in b{\cal E}\pare{\Z,\C}$. Indeed, we have
\begin{align*}
\E^{m_0}(f(S_j)) &= \sum_{m\in \Z}f(m)\Prob^{m_0}\pare{S_{j} = m}\\
&=f(0)\Prob^{m_0}\pare{S_{j} = 0} + \sum_{m\in \N^\ast}f(m)\Prob^{m_0}\pare{S_{j} = m} + \sum_{m\in \N^\ast}f(-m)\Prob^{m_0}\pare{S_{j} = -m}.
\end{align*}
Now using \eqref{eq:loi-classique-remplace-pseudo} (and the definition of killing), we find
\begin{align*}
\E^{m_0}(f(S_j))&=f(0)\Prob^{m_0}\pare{|M_{j}| = 0} + \alpha\sum_{m\in \N^\ast}f(m)\Prob^{m_0}\pare{|M_{j}| = m} + (1-\alpha)\sum_{m\in \N^\ast}f(-m)\Prob^{m_0}\pare{|M_{j}| = m}\\
&\hspace{0,4 cm}+ (1-\alpha)\pare{\sum_{m\in \N^\ast}f(m)\Prob^{m_0}\pare{|M^\dagger_{j}| = m} - \sum_{m\in \N^\ast}f(-m)\Prob^{m_0}\pare{|M_j^\dagger| = m}}\\
&= T_j^{0}f(m_0).
\end{align*}

The same kind of arguments may be invoqued for negative $m_0$ and the result of the lemma follows.
\end{proof}

\medskip

Analogously to the above, let us introduce $(W_t)_{t\geq 0}$ a standard Brownian motion constructed on $(\Omega, {\cal F}, {\Prob}^{x})$ starting from $x$ ({\it i.e.} $W_0=x$, ${\Prob}^{x}$-a.s).

Let $b{\cal E}\pare{\R,\C}$ stand for the set of Borel bounded complex valued functions defined on $\R$. Similarly as before, we introduce the following family of operators $\pare{T_t}_{t\geq 0}$ acting from $b{\cal E}\pare{\R,\C}$ to $b{\cal E}\pare{\R,\C}$ and defined by
\begin{align}
T_tf~:~x\mapsto 
\begin{cases}
\alpha \E^{x}f(|W_t|) + (1-\alpha)\E^{x}f(-|W_t|) + (1-\alpha)\pare{\E^{x}f(|W^\dagger_t|) - \E^{x}f(-|W^\dagger_t|)}\text{if }x\geq 0\\
\alpha \E^{|x|}f(|W_t|) + (1-\alpha)\E^{|x|}f(-|W_t|) + \alpha\pare{\E^{|x|}f(-|W^\dagger_t|) - \E^{|x|}f(|W^\dagger_t|)}\;\;\text{if }x< 0.
\end{cases}
\end{align}

\begin{lemma}
\label{lemma:semi-group-psdosbm}
For any $f\in b{\cal E}\pare{\R,\C}$,  
\[
\E^{x}(f(X_t)) = T_tf(x).
\]
\end{lemma}
\begin{proof}
Suppose $x>0$. We use \eqref{eq:walsh3} and check at once that
\begin{align}
\label{eq:TF-SBM}
\E^{x}(f(X_t))&=\int_{-\infty}^{+\infty} f(y){p}(t,x,y)dy\nonumber\\
&=\alpha\int_0^{+\infty} f(y)\hat{g}(t,x,y)dy + (1-\alpha)\int_0^{+\infty} f(y)\check{g}(t,x,y)dy\nonumber\\
&\hspace{0,5 cm} + \pare{1-\alpha}\pare{\int_{-\infty}^0 f(y)\hat{g}(t,x,-y)dy - \int_{-\infty}^{0}f(y)\check{g}(t,x,-y)dy}\nonumber\\
&=\alpha\int_0^{+\infty} f(y)\hat{g}(t,x,y)dy + (1-\alpha)\int_0^{+\infty} f(y)\check{g}(t,x,y)dy\nonumber\\
&\hspace{0,5 cm} + \pare{1-\alpha}\pare{\int_{0}^{+\infty} f(-y)\hat{g}(t,x,y)dy - \int_{0}^{+\infty} f(-y)\check{g}(t,x,y)dy}\nonumber\\
&=T_tf(x).
\end{align}
The same type of arguments may be invoqued for negative $x$ and the result of the lemma follows.
\end{proof}

The main result of this section is the following statement.
\begin{proposition}
\label{prop:conv-pseudo-MA}
The rescaled asymmetric random walk
converges in the sense of Definition~\ref{def_pseudo_weak_cv}
to the pseudo SBM~:
\[
{X}^n:=\pare{n^{-1/2}S_{\lfloor nt\rfloor}}_{t\geq 0}\stackrel{\text{pseudo-w}}{\longrightarrow} \pare{X_t}_{t\geq 0}\hspace{0.4 cm}\text{ as }\hspace{0.4 cm}n\rightarrow +\infty.
\]
\end{proposition}

\begin{proof}

Our main concern is to ensure that all arguments for the convergence in the classical case still hold true in our pseudo probability context. Since we deal with pseudo stochastic processes, the difficulties are two fold:
\begin{itemize}[label = $-$]
\item we are not allowed to use a Skorokhod embedding, a tool that is often used for the convergence of random walks;
\item we are not allowed to perform transformations directly on the path space such as reflection or killing for pseudo processes.
\end{itemize}
Note also that, except in the case $\alpha=1/2$, the skewed random walk and the skewed Brownian motion (classical or pseudo) do not have independent increments.
Following \cite{harrison-shepp} in the classical case where $\alpha\in (0,1)$, the idea is to observe the similarity of structure shared by the family of operators $(T_j^0)$ and $(T_t)$ (whose definition involve only classical processes) and take advantage of known results concerning the convergence of random walks. 

We only give the main arguments of the proof and leave the computational details to the reader.
\medskip

Let us fix $x\geq 0$, $t\geq 0$, $u\in \R$ and set $m_n :=\lfloor \sqrt{n}x \rfloor$. 

As mentionned above, in the whole proof
\begin{equation*}
\hat\psi_{0}(j;m_n,u) := \E^{m_n}\pare{\exp\pare{i u\:|M_{j}|}},\;\;\;\text{resp. }\;\;\check{\psi}_{0}(j;m_n,u) := \E^{m_n}(\exp(i u |M^{\dagger}_{j}|))
\end{equation*}
stands for the characteristic function of the classical standard symmetric random walk on $\Z$ reflected at $0$ (resp. killed at $0$) starting from $\lfloor \sqrt{n}x \rfloor$ and taken at time $j$ and evaluated at $u$.

Similarly
\begin{equation}
\label{eq:characteristic-BM}
\hat{\psi}(t;x,u) := \E^{x}({\rm e}^{iu|W_t|}),\;\;\;\;\;\text{resp.}\;\;\check{\psi}(t;x,u) := \E^{x}({\rm e}^{iu|{W}^\dagger_t|})
\end{equation}
stands for the characteristic function of a standard Brownian motion $(W_s)$ reflected at $0$ on $[0,\infty)$ (resp. killed at $0$) at time $t$ and evaluated at $u$.

Let us denote $h_u\in b{\cal E}\pare{\R,\C}$ the function defined by
\[
h_u~:~z\mapsto \exp(iu\,z).
\]

From the result of Lemma \ref{lemma:semi-group-psdosbm} we infer that
\begin{align}
\label{eq:decomp-TF-SBM}
\E^x\pare{\exp\pare{i u\;X_t}} &= T_th_u(x)\nonumber\\
&=\alpha\hat{\psi}(t;x,u) + (1-\alpha)\hat{\psi}(t;x,-u) + (1-\alpha)\pare{\check{\psi}(t;x,u) - \check{\psi}(t;x,-u)}
\end{align}

Similarly, from  the result of Lemma \ref{lemma:semi-group-srw}, we have
\begin{align}
\label{eq:decomp-TF-RW}
\E^{m_n}\pare{\exp\pare{i u\;n^{-1/2}S_{\lfloor nt\rfloor}}} &= T^0_{\lfloor nt\rfloor}h_{n^{-1/2}u}(m_n)\nonumber\\
&=\alpha\hat{\psi}_{0}(\lfloor nt\rfloor;m_n, n^{-1/2}u) + (1-\alpha)\hat{\psi}_{0}(\lfloor nt\rfloor;m_n,-n^{-1/2}u)\nonumber\\
&\hspace{0,5 cm} +  (1-\alpha)\pare{\check{\psi}_{0}(\lfloor nt\rfloor;m_n,n^{-1/2}u) -\check{\psi}_{0}(\lfloor nt\rfloor;m_n,-n^{-1/2}u)}.
\end{align}

Classical results regarding the convergence of characteristic functions for normalized symmetrized random walks reflected at $0$ (resp. killed at $0$) towards the characteristic functions of the standard Brownian motion reflected at $0$ (resp. killed at $0$) ensure that 
$\hat{\psi}_0\pare{\lfloor nt\rfloor;m_n,n^{-1/2}u}$ 
(resp. $\check{\psi}_0\pare{\lfloor nt\rfloor;m_n,n^{-1/2}u}$) 
converges simply to $\hat{\psi}(t; x, u)$ (resp. to $\check{\psi}(t; x, u)$) 
as $n$ tends to infinity. 

By comparison of \eqref{eq:decomp-TF-RW} with \eqref{eq:decomp-TF-SBM} and by the linearity of the convergence, we get
\[
\E^{m_n}\croc{\exp\pare{i u\,\,n^{-1/2}S_{\lfloor nt\rfloor}}}\xrightarrow[n \rightarrow +\infty]{}\E^{x}\croc{\exp\pare{i u X_t}}.\]
Similar arguments show that this convergence also holds for $x<0$ and $x=0$. Recasting this convergence using our operators writes: for any $x\in \R$,
\begin{equation}
\label{convergence-standard}
T^0_{\lfloor nt\rfloor}h_{n^{-1/2}u}(m_n)\xrightarrow[n \rightarrow +\infty]{}T_th_u(x).
\end{equation}
We leave it to the reader to check using very close arguments that we have also
\begin{equation}
\label{convergence-standard-2}
T^0_{\lfloor nt\rfloor + 1}h_{n^{-1/2}u}(m_n)\xrightarrow[n \rightarrow +\infty]{}T_th_u(x).
\end{equation}
\vspace{0,2 cm}

Let us now turn to the convergence of the characteristic function for joint pseudo distributions.

As above, let $x\geq 0$, $0\leq t_1<t_2$, $u_1,u_2\in \R$ and set $m_n =\lfloor \sqrt{n}x \rfloor$. Applying the Markov property of the pseudo random walk and the result of Lemma \ref{lemma:semi-group-srw} gives
\begin{align}
\label{eq:decomp-TF-RW-2-0}
&\Psi_{0}(\lfloor nt_1\rfloor,\lfloor nt_2\rfloor;m_n,u_1,u_2)\nonumber\\ 
&\hspace{0,5 cm}:=\E^{m_n}\croc{\exp\pare{i u_1\;n^{-1/2}S_{\lfloor nt_1\rfloor} + i u_2\;n^{-1/2}S_{\lfloor nt_2\rfloor}}}\nonumber\\ 
&\hspace{0,5 cm}=\E^{m_n}\croc{\exp\pare{i u_1\;n^{-1/2}S_{\lfloor nt_1\rfloor}}\E^{S_{\lfloor nt_1\rfloor}}\croc{\exp\pare{i u_2\;n^{-1/2}S_{\lfloor nt_2\rfloor - \lfloor nt_1\rfloor}}}}\nonumber\\
&\hspace{0,5 cm}=T^{0}_{\lfloor nt_1\rfloor}\croc{j\mapsto h_{n^{-1/2}u_1}(j)T^{0}_{\lfloor nt_2\rfloor - \lfloor nt_1\rfloor}h_{n^{-1/2}u_2}(j)}(m_n).
\end{align}
\medskip

Note that $(\lfloor nt_2\rfloor - \lfloor nt_1\rfloor) \in \{\lfloor n(t_2 - t_1)\rfloor, \lfloor n(t_2 - t_1)\rfloor + 1\}$. As before, we are in position to use standard convergence results regarding the characteristic functions of normalized classical symmetric random walks.

Using \eqref{convergence-standard} and \eqref{convergence-standard-2}, we prove that $\Psi_{0}(\lfloor nt_1\rfloor,\lfloor nt_2\rfloor;m_n,u_1,u_2)$ converges as $n$ tends to infinity to
\begin{align*}
\Psi(t_1,t_2;x,u_1,u_2) &:= T_{t_1}\croc{y\mapsto h_{u_1}(y)T_{t_2 - t_1}h_{u_2}(y)}(x).
\end{align*}
Note that from the definitions of the operators $(T^0_j)_{j\in \N}$ and $(T_t)_{t\geq 0}$ this convergence results from the properties of the weak convergence of classical symmetric random walks (reflected or killed) towards their corresponding standard Brownian motion (reflected or killed). 


It remains to check that 
\begin{equation}
\label{eq:final}
\Psi(t_1,t_2;x,u_1,u_2) = \E^x\croc{\exp\pare{i (u_1 X_{t_1} + u_2 X_{t_2}}}
\end{equation}
in order to conclude our proof. But, exactly as before this last equality is easily verified by the use of the Markov property for the pseudo-process $(X_t)$ at time $t_1$ (which itself derives directly from the Chapman-Kolmogorov identity). Indeed, by application of Lemma \ref{lemma:semi-group-psdosbm} and the definition of $h_{u}$, we have
\begin{align*}
\Psi(t_1,t_2;x,u_1,u_2) &:= T_{t_1}\croc{y\mapsto h_{u_1}(y)T_{t_2 - t_1}h_{u_2}(y)}(x)\\
&=\E^x\croc{\exp\pare{iu_1\,X_{t_1}}\E^{X_{t_1}}\pare{\exp\pare{iu_2\,X_{t_2 - t_1}}}}\\
&=\E^x\croc{\exp\pare{i(u_1X_{t_1} + u_2X_{t_2}}}
\end{align*}
which gives \eqref{eq:final}.

Hence, 
\begin{align*}
\lim_{n\rightarrow +\infty}\E^{m_n}\croc{\exp\pare{i\pare{u_1\,n^{-1/2}S_{\lfloor n t_1\rfloor} + u_2\,n^{-1/2}S_{\lfloor n t_2\rfloor}}}}=\E^x\croc{\exp\pare{i\pare{u_1\,X_{t_1}+ u_2\,X_{t_2}}}}.
\end{align*}

\medskip
\medskip

Finally, we may extend the previous limit result by induction to prove that
for any $\ell\in \N^{\ast}$, $u_1,\dots, u_\ell\in \R$ and times $0\leq t_1 <\dots < t_\ell$,
\[\lim_{n\rightarrow +\infty}\E^{m_n}\croc{\exp\pare{i\sum_{j=1}^\ell u_j\,n^{-1/2}S_{\lfloor n t_j\rfloor}}}= \E^x\croc{\exp\pare{i\sum_{j=1}^\ell u_jX_{t_j}}}.
\]
The proof is completed.
\end{proof}

\section{The evolution equation on a finite interval}
\label{sec:fund-sol}

In this section we assume $0<a<\infty$.

Again, we first address the case $k > 0$, in which we compute the transition function of the skew Brownian motion killed at the end-points of $I=(-a,a)$, then we treat the general case $k\in\R^*\setminus\{-1\}$.

\subsection{The transition function of the SBM killed at $-a$ or $a$ (case $k>0$)}
\label{ssec:SBMkilled}

In the same manner as for \eqref{eq:FK1}, it can be shown that 
the solution to the parabolic problem~\eqref{eq_evol} on $(-a,a) \times \R$
can be represented as
\begin{equation}
\label{eq:FK2}
u(t,x)=\E^x[u_0(\check{X}_t)]
\end{equation}
where $\check{X}$ is the SBM of parameter $\beta=\frac{1-k}{1+k}$, killed at $-a$ or $a$. 
This process behaves like the SBM $X$ as long as it does not exit from $(-a,a)$. 
When it touches $-a$ or $a$ it is sent at a cemetery point $\partial$. 
By convention, for any function $f$ one has $f(\partial)=0$, 
which ensures that the homogeneous Dirichlet boundary condition in \eqref{eq:CL-evol}
is satisfied. 
\medskip

Let us assume that $x \in (-a,a)$ and let 
$\check{T}_{(-a,a)}=\inf\{t\geq 0:\,\check{X}_t\notin (-a,a)\}$. 
We compute
\begin{equation}
\label{eq:FK3}
\begin{array}{lll}
\E^x[u_0(\check{X}_t)]&=&\E^x[u_0(\check{X}_t);\, \check{T}_{(-a,a)} \leq t]+\E^x[u_0(\check{X}_t);\,\check{T}_{(-a,a)}> t]\\
\\
&=&\E^x[u_0(\partial);\, \check{T}_{(-a,a)} \leq t]+\E^x[u_0(\check{X}_t);\,\check{T}_{(-a,a)}> t]\\
\\
&=&\E^x[u_0(X_t);\,T_{(-a,a)}> t],
\end{array}
\end{equation}
where $X$ is the SBM considered in Section \ref{ssec:pseudo-SBM} and $T_{(-a,a)}=\inf\{t\geq 0:\,X_t\notin (-a,a)\}$.

We first derive the expression of the kernel $\check{p}(t,x,y)$ 
that satisfies $\P^x(X_t\in dy;\,T_{(-a,a)}>t)=\check{p}(t,x,y)dy$, 
i.e. the fundamental solution to \eqref{eq_evol}~:

\begin{proposition}
\label{prop:trans-Xkilled}
Let $\beta\in (-1,1)$. Let $X$ be the solution of \eqref{eq:SBM}, i.e. the SBM of parameter $\beta$.
Let $a>0$.
 
For any $x,y\in(-a,a)$, any $t> 0$, one has $\P^x(X_t\in dy;\,T_{(-a,a)}>t)=\check{p}(t,x,y)dy$ with
\begin{equation}
\label{eq:trans}
\check{p}(t,x,y)=\left\{
\begin{array}{lll}
(1-\beta)p^{(-a,a)}_W(t,x,y)&\text{if}&x\geq 0,\,y<0\\
\\
-\beta p^{(0,a)}_W(t,x,y)+(1+\beta)p^{(-a,a)}_W(t,x,y)&\text{if}&x>0,\,y>0\\
\\
(1+\beta)p^{(-a,a)}_W(t,x,y)&\text{if}&x\leq 0,\,y>0\\
\\
\beta p^{(-a,0)}_W(t,x,y)+(1-\beta)p^{(-a,a)}_W(t,x,y)&\text{if}&x<0,\,y<0.\\
\end{array}
\right.
\end{equation}
\end{proposition}

\begin{proof}[Proof of Proposition \ref{prop:trans-Xkilled}]
We use Lemma \ref{lem:walsh} and take advantage of the symmetry of the space interval $(-a,a)$: then it is possible to adapt the arguments in \cite{walsh}, using for example in particular the transition function of  the Brownian motion killed at $-a$ or $a$ instead of the transition function of the Brownian motion. We detail all the steps for the sake of completeness.

We set $T_0=\inf\{t\geq 0:\,X_t=0\}$ and 
remark that $T_{(-a,a)}=\inf\{t\geq 0:\,X_t\notin (-a,a)\}=\inf\{t\geq 0:\,|X_t|=a\}$. 
We denote~
$\tau_0=\inf\{t\geq 0:\,W_t=0\}$ and $\tau_{(-a,a)}=\inf\{t\geq 0:\,|W_t|=a\}$. 
In the computations below, by a slight abuse of notation we may denote by $\P^x$ 
either $\P(\cdot|X_0=x)$ or $\P(\cdot|W_0=x)$. 
This will be clear from the context.

Let $t>0$.
We  first treat the case $x>0$, $y<0$. 
Using the fact that $\P^x(X_t\in dy; T_0> t)=0$, and using the strong Markov property of $X$, 
it follows that
$$
\begin{array}{l}
\P^x(X_t\in dy;\,T_{(-a,a)}>t)\\
\\
=\P^x(X_t\in dy; T_0\leq t;\,T_{(-a,a)}>t)+\P^x(X_t\in dy; T_0> t;\,T_{(-a,a)}>t)\\
\\
=\E^x\big[ \1_{T_0\leq t} \P^x(X_t\in dy;\,T_{(-a,a)}>t|\cF_{T_0})\big]\\
\\
=\int_0^t\P^0(\tilde{X}_{t-u}\in dy;\,\tilde{T}_{(-a,a)}>t-u)f_{T_0}^x(u)du,\\
\end{array}
$$
where $\tilde{X}$ is another SBM of parameter $\beta$, starting from zero under $\P^0$, $\tilde{T}_{(-a,a)}=\inf\{t\geq 0:\,|\tilde{X}_t|=a\}$ and $f_{T_0}^x(u)du$ is the law of 
$T_0$ under $\P^x$.

From Lemma \ref{lem:walsh} we have
$$
\begin{array}{l}
\P^0(\tilde{X}_{t-u}\in dy;\,\tilde{T}_{(-a,a)}>t-u)\\
\\
=\P^0(|\tilde{X}_{t-u}|\in -dy;\,\tilde{T}_{(-a,a)}>t-u; \tilde{X}_{t-u}<0)\\
\\
=\frac{1-\beta}{2}\P^0(|\tilde{W}_{t-u}|\in -dy;\,\tilde{\tau}_{(-a,a)}>t-u)\\
\\
=\frac{1-\beta}{2}\Big( \P^0(\tilde{W}_{t-u}\in -dy;\,\tilde{\tau}_{(-a,a)}>t-u) + \P^0(-\tilde{W}_{t-u}\in -dy;\,\tilde{\tau}_{(-a,a)}>t-u)\Big)\\
\\
=(1-\beta)\P^0(\tilde{W}_{t-u}\in dy;\,\tilde{\tau}_{(-a,a)}>t-u)
\end{array}
$$
where $\tilde{W}$ is a Brownian motion starting from zero under $\P^0$ and $\tilde{\tau}_{(-a,a)}=\inf\{t\geq 0:\,|\tilde{W}_t|=a\}$.

Thus denoting by $f_{\tau_0}^x(u)du$  the law of 
$\tau_0$ under $\P^x$, noticing that $f_{T_0}^x(u)=f_{\tau_0}^x(u)$, 
and using this time the strong Markov property of $W$, we calculate
$$
\begin{array}{l}
\P^x(X_t\in dy;\,T_{(-a,a)}>t)\\
\\
=(1-\beta)\int_0^t\P^0(\tilde{W}_{t-u}\in dy;\,\tilde{\tau}_{(-a,a)}>t-u)f_{\tau_0}^x(u)du,\\
\\
=(1-\beta)\P^x(W_t\in dy; \tau_0\leq t;\,\tau_{(-a,a)}>t)\\
\\
=(1-\beta)\P^x(W_t\in dy; \,\tau_{(-a,a)}>t)\\
\\
=(1-\beta)p^{(-a,a)}_W(t,x,y)dy,
\end{array}
$$
and obtain the first line of \eqref{eq:trans}.

Next, we treat the case $x>0$ and $y>0$. We have
$$
\begin{array}{l}
\P^x(X_t\in dy;\,T_{(-a,a)}>t)\\
\\
=\P^x(X_t\in dy; T_0\leq t;\,T_{(-a,a)}>t)+\P^x(X_t\in dy; T_0> t;\,T_{(-a,a)}>t)\\
\end{array}
$$
As $0<x,y<a$, the term $\P^x(X_t\in dy; T_0> t;\,T_{(-a,a)}>t)$ corresponds to the transition of a Brownian motion killed at $0$ or $a$, or in other words,
$$
\P^x(X_t\in dy; T_0> t;\,T_{(-a,a)}>t)=p^{(0,a)}_W(t,x,y)dy.$$
As for the term $\P^x(X_t\in dy; T_0\leq t;\,T_{(-a,a)}>t)$, 
we use similar computations as in the case $x>0$, $y<0$. 
Since $X_t$ is positive this time, we obtain
$$
\P^x(X_t\in dy; T_0\leq t;\,T_{(-a,a)}>t)=(1+\beta)\P^x(W_t\in dy; \tau_0\leq t;\,\tau_{(-a,a)}>t).$$
Notice also that
$$
\begin{array}{l}
\P^x(W_t\in dy; \tau_0\leq t;\,\tau_{(-a,a)}>t)\\
\\
=\P^x(W_t\in dy; \,\tau_{(-a,a)}>t)-\P^x(W_t\in dy; \tau_0> t;\,\tau_{(-a,a)}>t)\\
\\
=p^{(-a,a)}_W(t,x,y)dy-p^{(0,a)}_W(t,x,y)dy.\\
\end{array}
$$
Putting all the pieces together yields
$$
\begin{array}{lll}
\P^x(X_t\in dy;\,T_{(-a,a)}>t)&=&\Big[  p^{(0,a)}_W(t,x,y) + (1+\beta) \{p^{(-a,a)}_W(t,x,y)-p^{(0,a)}_W(t,x,y)\}\Big]dy\\
\\
&=&\Big[-\beta p^{(0,a)}_W(t,x,y)+(1+\beta)p^{(-a,a)}_W(t,x,y)\Big]dy,\\
\end{array}
$$
and thus, the second line of \eqref{eq:trans}. 
The remaining cases can be treated in a similar manner.
\end{proof}

\begin{remark}
\label{rem:coherent-walsh}
Note the consistence of \eqref{eq:walsh2} with \eqref{eq:trans}, as $a\to\infty$ 
in \eqref{eq:trans}.
\end{remark}

\begin{remark}
\label{rem:int-1}
Note that as
$$
\int_{-a}^a\check{p}(t,x,y)dy=\P^x(\,T_{(-a,a)}>t),
$$
$\check{p}(t,x,\cdot)$ does not necessarily integrate to 1. 
The definition of a genuine family of transition probability measures
from $\check{p}(t,x,y)dy$ 
would require combining this density measure with a Dirac measure that charges $\partial$ 
with probability $\P^x(\,T_{(-a,a)}\leq t)$, see \cite{RY}, p84.

Still by a slight abuse of language we call this kernel 
the transition function of $\check{X}$. 

Note that in the hereafter examined case of a negative coefficient, the integrals of $\check{p}(t,x,\cdot)$ do not even necessarily correspond to true probabilities.
\end{remark}

\subsection{The fundamental solution to~\eqref{eq_evol} in the general case $k\in\R^* \setminus \{-1\}$} 

When $k > 0$, the transition function yields a fundamental solution 
to the evolution equation. When $k < 0, k \neq -1$, as in the case when $I = \R$, we check that 
its expression also yields a fundamental solution.

\begin{lemma}
\label{lem:verif}
Let $k\in\R^* \setminus \{-1\}$ and define $\beta=\frac{1-k}{1+k}$. 
The kernel $\check{p}(t,x,y)$ defined by \eqref{eq:trans} is the fundamental solution to \eqref{eq:evol-nondiv}-\eqref{eq:CL-evol}, or to \eqref{eq_evol}.
Equivalently it satisfies for any $y\in(-a,a)$,
\begin{equation}
\label{eq:fund-lap}
\partial_t\check{p}(t,x,y)=\frac12\partial^2_{xx}\check{p}(t,x,y),\quad \forall x\in(-a,0)\cup(0,a),
\end{equation}
 and
\begin{equation}
\label{eq:fund-trans-trace}
\check{p}(t,0-,y)=\check{p}(t,0+,y),
\end{equation}
\begin{equation}
\label{eq:fund-trans}
k \partial_{x}\check{p}(t,0-,y)=\partial_{x}\check{p}(t,0+,y),
\end{equation}
and
\begin{equation}
\label{eq:fund-lim}
\lim_{x\to \pm a}\check{p}(t,x,y)=0.
\end{equation}
\end{lemma}

\begin{proof}
If a kernel $q(t,x,y)$ is a fundamental solution to \eqref{eq:evol-nondiv}-\eqref{eq:CL-evol},
then by definition, the function
\begin{equation}
\label{eq:eq-fund-bis}
(t,x)\mapsto u(t,x) = \int_{-a}^au_0(y)q(t,x,y)dy
\end{equation}
solves  \eqref{eq:evol-nondiv}-\eqref{eq:CL-evol} for any $u_0\in C^\infty_c(-a,a)$. As 
\[ \begin{array}{rcllcl}
\partial_tu(t,x) &=& \ds\int_{-a}^au_0(y)\partial_tq(t,x,y)dy,
&\quad
 \partial_xu(t,x) &=& \ds\int_{-a}^au_0(y)\partial_xq(t,x,y)dy,
 \\
\textrm{and}\quad
\partial^2_{xx}u(t,x) &=&\ds\int_{-a}^au_0(y)\partial^2_{xx}q(t,x,y)dy,
\end{array}
\]
letting $u_0$ vary in  $C^\infty_c(-a,a)$, it is easy to see that
$q(t,x,y)$ satisfies \eqref{eq:fund-lap}-\eqref{eq:fund-lim}. Conversely,
if \eqref{eq:fund-lap}-\eqref{eq:fund-lim} hold for $q(t,x,y)$, one may check that the function defined by \eqref{eq:eq-fund-bis} solves \eqref{eq:evol-nondiv}-\eqref{eq:CL-evol}. Thus the equivalence.
\medskip

Consider first $k>0$ and $\beta=\frac{1-k}{1+k}$. From \eqref{eq:FK2}, \eqref{eq:FK3} and Proposition \ref{prop:trans-Xkilled}, it is clear that the kernel
$\check{p}(t,x,y)$ defined by \eqref{eq:trans} is the fundamental solution to \eqref{eq:evol-nondiv}-\eqref{eq:CL-evol}, and therefore
satisfies  \eqref{eq:fund-lap}-\eqref{eq:fund-lim}.

Considering \eqref{eq:fund-lap} and for example the first line of \eqref{eq:trans} it is clear that for $x<0$ and $y>0$ one has 
\begin{equation}
\label{eq:chal-paa}
\partial_tp^{(-a,a)}_W(t,x,y)=\frac12\partial^2_{xx}p^{(-a,a)}_W(t,x,y)
\end{equation}
(of course this fact can formally be seen from  \eqref{eq:Bro-killed}; see also Remark~\ref{rem:form} below).
\medskip

Consider now $k\in\R_-^*\setminus\{-1\}$, set $\beta=\frac{1-k}{1+k}$, and
define the kernel $\check{p}(t,x,y)$ by~\eqref{eq:trans}.
From the first line of \eqref{eq:trans}, and \eqref{eq:chal-paa} it is clear that \eqref{eq:fund-lap} holds for $x<0$ and $y>0$. The other cases can be treated similarly, so that 
$\check{p}(t,x,y)$ solves \eqref{eq:fund-lap} for any $y\in(-a,a)$, $x\in(-a,0)\cup(0,a)$.

Concerning the transmission conditions, we first note that when $y > 0$, 
Eq. \eqref{eq:fund-trans} can be rewritten
\begin{equation}
\label{eq:cond-trans-beta}
(1+\beta)\partial_{x}\check{p}(t,0+,y)=(1-\beta)\partial_{x}\check{p}(t,0-,y)
\end{equation}
or equivalently, using the definition of $\check{p}(t,x,y)$ in~\eqref{eq:trans}
(with $y < 0$)
\begin{equation*}
(1-\beta^2)\partial_xp_W^{(-a,a)}(t,0-,y)=
(-\beta-\beta^2)\partial_xp_W^{(0,a)}(t,0+,y)+
(1+2\beta+\beta^2)\partial_xp_W^{(-a,a)}(t,0+,y),
\end{equation*}
which in turn, due to the continuity of $\partial_xp_W^{(-a,a)}(t,0,y)$ at $x=0$, 
is equivalent to
\begin{equation}
\label{eq:depend-pas-beta}
\partial_xp_W^{(0,a)}(t,0+,y)=2\partial_xp_W^{(-a,a)}(t,0,y).
\end{equation}
Note that condition  \eqref{eq:depend-pas-beta} only bears on the properties of $p_W^{(-a,a)}$ and $p_W^{(0,a)}$, and is equivalent to
\eqref{eq:fund-trans} whatever
 the sign of $k$, provided~$k \neq -1$. Thus we know from the case $k > 0$ that \eqref{eq:depend-pas-beta} holds, and thus \eqref{eq:fund-trans} must hold for all $k \neq -1$. 
The same argument applies to $y<0$.
Equation \eqref{eq:fund-trans-trace} is a consequence of~\eqref{eq:Bro-killed-lim} 
and of the continuity at $x=0$ of the kernel~$p_W^{(-a,a)}(t,x,y)$. 

Finally, the Dirichlet condition \eqref{eq:fund-lim} follows from \eqref{eq:Bro-killed-lim} 
and the form of $\check{p}(t,x,y)$.
We conclude that the latter is indeed
the fundamental solution of \eqref{eq:evol-nondiv}-\eqref{eq:CL-evol} when $k < 0, k \neq -1$.
\end{proof}


\begin{remark}
\label{rem:form}
In fact it is possible to start from \eqref{eq:trans}, and to use \eqref{eq:Bro-killed}  in order to formally check~\eqref{eq:fund-lap}-\eqref{eq:fund-trans}, for any $k\in\R^*\setminus\{-1\}$. We especially want to explain how one can derive the transmission condition in \eqref{eq:fund-trans}.

As we have already seen in  the proof of Lemma \ref{lem:verif}, it is enough to check \eqref{eq:depend-pas-beta}. On one hand one has
$$
\begin{array}{l}
\partial_xp_W^{(0,a)}(t,0+,y)\\
\\
=\ds\sum_{n=-\infty}^\infty\big\{  \frac{y-2na}{t}g(t,0,y-2na)-\frac{-y-2na}{t}g(t,0,-y-2na)  \big\} \\
\\
=\ds \sum_{n=-\infty}^\infty\big\{  \frac{y-2na}{t}g(t,0,y-2na)+\frac{y+2na}{t}g(t,0,y+2na)  \big\} \\
\\
=\ds 2 \sum_{n=-\infty}^\infty \frac{y+2na}{t}g(t,0,y+2na)\\
\end{array}
$$
On the other hand one has
$$
\begin{array}{l}
\partial_xp_W^{(-a,a)}(t,0,y)\\
\\
=\ds\sum_{n=-\infty}^\infty\big\{  \frac{y-4na}{t}g(t,0,y-4na)-\frac{2a-y-4na}{t}g(t,0,2a-y-4na)  \big\} \\
\\
=\ds \sum_{n=-\infty}^\infty\big\{  \frac{y-4na}{t}g(t,0,y-4na)+\frac{y+4na-2a}{t}g(t,0,y+4na-2a)  \big\} \\
\\
=\ds \sum_{n=-\infty}^\infty\big\{  \frac{y-2(2n)a}{t}g(t,0,y-2(2n)a)+\frac{y+2(2n-1)a}{t}g(t,0,y+2(2n-1)a)  \big\} \\
\\
=\ds  \sum_{n=-\infty}^\infty \frac{y+2na}{t}g(t,0,y+2na).\\
\end{array}
$$
Therefore \eqref{eq:depend-pas-beta}. 
\end{remark}

Note that in the case $k\in \R_-^*\setminus\{-1\}$ the kernel $\check{p}(t,x,y)$ defined by \eqref{eq:trans} is not positive: indeed in that case $\beta>1$ and it suffices to examine the first line of \eqref{eq:trans}.

So even if we complement the density measure $\check{p}(t,x,y)dy$ with Dirac measures in order to obtain a family of transition pseudo probability measures (see our Remark \ref{rem:int-1}), the latter will not define a Markov process.
One could however introduce the pseudo-process associated to this family, in the spirit of 
Section~\ref{ssec:pseudo-SBM}.

\begin{remark} \label{rem:spectral_form_kernel}
We end up this section by noticing that as expected, the form of the kernel~\eqref{eq:trans} coincides with~\eqref{Green_spect}.
Indeed, the transition density of the Brownian motion on an interval $(a,b)$,
killed at $a$ or $b$ 
has the following spectral representation 
(see Appendix~1 in~\cite{borodin})
\begin{eqnarray*}
p^{(a,b)}_W(t,x,y) &=&
\ds\frac{2}{b-a}
\sum_{n \geq 1} 
\textrm{exp}\Big(-\ds\frac{n^2 \pi^2}{2(b-a)^2} t\Big)
\sin\Big(\ds\frac{n\pi}{b-a}(x-a)\Big) \sin\Big(\ds\frac{n\pi}{b-a}(y-a)\Big).
\end{eqnarray*}
Thus one has
\begin{eqnarray*}
p^{(-a,a)}_W(t,x,y)
&=&
\ds\frac{1}{a}
\sum_{n \geq 1} 
\textrm{exp}\Big( -\ds\frac{n^2 \pi^2}{8a^2} t\Big)
\sin\Big(\ds\frac{n\pi}{2a}(x-a)\Big) \sin\Big(\ds\frac{n\pi}{2a}(y-a)\Big).
\end{eqnarray*}
Regrouping the terms with odd and even indices, the above expression
rewrites, when $x\geq 0$ and $y<0$
\begin{eqnarray*}
p^{(-a,a)}_W(t,x,y)
&=&
+ \ds\frac{1}{a}
\sum_{q \geq 1} 
\textrm{exp}\Big( - \ds\frac{q^2 \pi^2}{2 a^2} t \Big)
\sin\Big(\ds\frac{q\pi}{a}(x-a)\Big) \sin\Big(\ds\frac{q\pi}{a}(y-a)\Big)
\\
&& 
+ \ds\frac{1}{a}
\sum_{q \geq 1} 
\textrm{exp}\Big( - \ds\frac{(2q-1)^2 \pi^2}{8a^2} t \Big)
\cos\Big(\ds\frac{(2q-1)\pi}{2a}(x-a)\Big) \cos\Big(\ds\frac{(2q-1)\pi}{2a}(y-a)\Big)
\\
&=&
\ds\frac{1}{a}
\sum_{q\geq 1} 
\Big[ \ds\frac{1}{k}e^{-\mu_q^2 t/2} \,g_q(x) g_q(y)
\;+\; e^{-\lambda_q^2 t/2} \, f_q(x) f_q(y) \Big]
\end{eqnarray*}
Multiplying by $1-\beta= \frac{2k}{k+1}$ to retrieve the first line of \eqref{eq:trans}, one easily recovers 
the expression~\eqref{Green_spect} when~$x \geq 0$ and $y < 0$. 
A similar calculation applies for the other possible choices of $x$ and $y$.
\end{remark}

\section{Various numerical schemes for the PDE \eqref{eq_evol}}
\label{sec:num}

In this section we want to construct several numerical schemes for the approximation of the solution $u$ of \eqref{eq_evol}, inspired by the theoretical results of the previous sections.

In Section \ref{ssec:num-spec} we construct a scheme $\bar{u}^N_{spec}$ inspired by the spectral representation of the semigroup of Section \ref{sec:spectral}.

In Section \ref{ssec:num-RW} we explain how we can infer a finite difference type scheme $\bar{u}^n_{RW}$ from the scaled pseudo asymmetric random walk $\widehat{X}^n$ of Section \ref{sssec:pseudo-RW}.

In Section \ref{ssec:num-kill-bro} we construct a scheme $\bar{u}^{h,N}_{fund}$, which is inspired by the fact that the fundamental solution $\check{p}(t,x,y)$ computed in Section \ref{ssec:SBMkilled} involves transition functions of killed Brownian motions, which can be seen as the fundamental solutions of simple heat equations, with homogeneous Dirichlet boundary conditions.

\vspace{0.2cm}
An initial condition $u_0$ is given (in Section \ref{ssec:num-spec} it is of class $L^2(I)$, in Sections  \ref{ssec:num-RW}  and \ref{ssec:num-kill-bro} we can imagine it is continuous and bounded).

We have $k\in\R^* \setminus \{-1\}$ and set $\beta=\frac{1-k}{1+k}$ and $\alpha=\frac{1+\beta}{2}$.

In the case $k\in (0,\infty)$  several numerical schemes are available (one can for example first perform a finite element discretization w.r.t. the space variable, and then a Crank-Nicholson scheme, e.g. \cite{raviart}), including probabilistic ones
(e.g. \cite{etore05a}, \cite{martinez12a}, \cite{etoremartinez1}). Thus, for this well explored case the schemes presented hereafter may provide additional methods (up to our knowledge they have never been proposed in this form).

Their main interest is that they allow to handle the case $k\in\R_-^* \setminus \{-1\}$, for which in particular classical probabilistic methods (e.g. \cite{etore05a}, \cite{martinez12a}, \cite{etoremartinez1}) cannot be applied. In the case $k>0$ the latter deeply rely on stochastic simulations of the trajectories of the SBM of parameter~$\beta\in(-1,1)$. When~$k\in\R_-^* \setminus \{-1\}$ one can define a pseudo-SBM, but it proves difficult to define associated trajectories. 
Note also that the techniques for proving the convergence of deterministic schemes do not apply in the case $k \in \R_-^* \setminus \{-1\}$.

\subsection{Scheme inspired by the spectral representation of the semi-group}
\label{ssec:num-spec}

We assume that $u_0\in L^2(I)$ and fix a truncation order $N\in\N^*$.  Then $\bar{u}^N_{spec}$ is defined by 
$$
\bar{u}^N_{spec}(t,x)= \sum_{n = 1} ^N
a_n e^{-\lambda_n^2t/2} f_n(x) + b_n e^{- \mu_n^2 t/2} g_n(x)
$$
where the  $\lambda_n$'s, $\mu_n$'s, $g_n$'s and  $f_n$'s are as in Section \ref{sec:spectral}, and the $a_n$'s and $b_n$'s are defined as in \eqref{eq:def-coeff-ab}. That is to say $\bar{u}^N_{spec}(t,x)$ is obtained by keeping the first $N$ terms in the spectral representation \eqref{eq_sol} of $u(t,x)$.

\vspace{0.2cm}
Of course in order to use this method we have to be able to compute the  $a_n$'s and $b_n$'s. In Section~\ref{sec:experiments} we consider two examples where these computations can be done explicitly. If this is not the case we can resort to numerical integration. We sum up hereafter the proposed algorithm (Algorithm 1).

\vspace{0.3cm}

\fbox{\parbox{\textwidth}{
{\bf ALGORITHM 1:} Computation of $\bar{u}_{spec}^N(t,x)$.

\vspace{0.2cm}
{\bf Parameters of the method:} A time $0\leq t<\infty$ at which we want to compute the approached solution.
 
 A truncation order $N\in\N^*$.
 
Remember that  the  $\lambda_n$'s, $\mu_n$'s, $g_n$'s and  $f_n$'s are as in Section \ref{sec:spectral}.

\vspace{0.3cm}
{\bf Algorithm:} 1) Compute
\begin{eqnarray*}
a_n \;=\; \ds\frac{ \ds\int_I A(x) u_0(x) f_n(x) \,dx}
{\ds\int_I A(x) |f_n(x)|^2 \,dx},
&\quad&
b_n \;=\; \ds\frac{ \ds\int_I A(x) u_0(x) g_n(x) \,dx}
{\ds\int_I A(x) |g_n(x)|^2 \,dx}
\end{eqnarray*}
(either exactly or by numerical integration).

\vspace{0.2cm}
   2) Return 
   $$
\bar{u}^N_{spec}(t,x)= \sum_{n = 1} ^N
a_n e^{-\lambda_n^2t/2} f_n(x) + b_n e^{- \mu_n^2 t/2} g_n(x).
$$
      }}

\vspace{0.3cm}
As the series in \eqref{eq_sol} is normally convergent we get immediately the following convergence result.

\begin{proposition}
Let us consider $u$ the solution of \eqref{eq_evol}. Let $t\geq 0$ and let us consider for any $N\in\N^*$ the function $\bar{u}^N_{spec}(t,\cdot)$ defined by Algorithm 1. We have
$$
||u(t,\cdot)-\bar{u}^N_{spec}(t,\cdot)||_{L^2(I)}\xrightarrow[N\to\infty]{}0.
$$
\end{proposition}

\subsection{Scheme inspired by the scaled pseudo asymmetric random walk}
\label{ssec:num-RW}

Let $u_0:\R\to\R$ be continuous and bounded and consider first \eqref{eq_cauchy}. Let us fix $n\in\N^*$ a discretization order.

From Proposition \ref{prop:conv-pseudo-MA}, $u(T,x)=\E^x[u_0(X_T)]$ is approached by $\E^x[u_0(\widehat{X}^n_T)]=\E^x[u_0(n^{-1/2}S_{\lfloor nT\rfloor})]$. We recall that the expectation symbols have to be understood as pseudo expectations, and that $X$ is the pseudo SBM of Section \ref{ssec:pseudo-SBM}.

Let $T>0$ and let us assume that $nT=N$, an integer. Denoting $u_0^n$ the function defined by $u_0^n(z)=u_0(n^{-1/2}\, z)$, $z\in\Z$, one has
\begin{equation*}
\E^x[u_0(\widehat{X}^n_T)]\approx \E^{\lfloor \sqrt{n}x\rfloor}[u_0^n(S_N)]
\end{equation*}
where the expectation in the right hand side is computed under $\P^{\lfloor \sqrt{n}x\rfloor}$ s.t. $\P^{\lfloor \sqrt{n}x\rfloor}(S_0=\lfloor \sqrt{n}x\rfloor)=1$.

Let us denote $v_N=u_0^n$. Using the (pseudo) Markov property of the (pseudo) random walk $S$ we get
\begin{equation}
\label{eq:conditioning1}
\E^{\lfloor \sqrt{n}x\rfloor}[u_0^n(S_N)]=\E^{\lfloor \sqrt{n}x\rfloor}\big[ \E[\,v_N(S_N)\,|\,S_{0},\ldots,S_{N-1}\,] \,\big]=\E^{\lfloor \sqrt{n}x\rfloor}\big[ \,v_{N-1}(S_{N-1}) \,\big]
\end{equation}
where we have denoted $v_{N-1}(z)=\E[\,v_N(S_N)\,|\,S_{N-1}=z\,]$, $z\in\Z$. In fact, defining more generally,
$$
v_{m-1}(z)=\E[\,v_m(S_m)\,|\,S_{m-1}=z\,], \quad z\in\Z,\quad 1\leq m\leq N,
$$
and proceeding to further conditionings in \eqref{eq:conditioning1} we get
$$
\E^{\lfloor \sqrt{n}x\rfloor}[u_0^n(S_N)]=\E^{\lfloor \sqrt{n}x\rfloor}\big[ \,v_{N-2}(S_{N-2}) \,\big]=\ldots=\E^{\lfloor \sqrt{n}x\rfloor}\big[ \,v_{0}(S_{0}) \,\big]=v_0(\lfloor \sqrt{n}x\rfloor).
$$
Note that from the (possibly pseudo) transition probabilities of the random walk $S$ we have for any $1\leq m\leq N$,
$$
v_{m-1}(z)=\frac 1 2 [v_m(z+1)+v_m(z-1)]\1_{z\neq 0}+ [\alpha v_m(z+1)+(1-\alpha)v_m(z-1)]\1_{z= 0}, \quad z\in\Z.
$$
To sum up, one may approach $\E^x[u_0(\widehat{X}^n_T)]$ by computing $v_0(\lfloor \sqrt{n}x\rfloor)$  through the dynamical programming procedure
\begin{equation}
\label{eq:algo-rec}
\left\{
\begin{array}{llll}
v_N(z)&=&u_0^n(z),&\forall z\in\Z\\
\\
v_{m-1}(z)&=&\dfrac 1 2 [v_m(z+1)+v_m(z-1)]\1_{z\neq 0}&\forall z\in\Z, \,\forall 1\leq m\leq N.\\
\\
&&\quad\quad+ [\alpha v_m(z+1)+(1-\alpha)v_m(z-1)]\1_{z= 0},&\\
\end{array}
\right.
\end{equation}
The algorithm \eqref{eq:algo-rec} is written in a recursive form. We  now rewrite it in an iterative form -using also a new set of notations, in order to stress the fact that it is very similar to an explicit finite difference  scheme, with space step 
$h=n^{-1/2}$ and time step $\delta t=n^{-1}$. 

Let us consider the space grid $\{x_j\}_{j\in\Z}$ defined by $x_j=j/\sqrt{n}$ for any $j\in\Z$, and the scheme $\{U_j^m\}$, for $j\in\Z$, $0\leq m\leq M$, defined by
\begin{equation}
\label{eq:CI-DIF1}
U_j^0=u_0^n(j)=u_0(x_j),\quad j\in\Z
\end{equation}
and, for $0\leq m\leq N-1$,
 \begin{eqnarray}
  \label{eq:DF1}
   U^{m+1}_j&=&\frac12 U^{m}_{j+1}+\frac 1 2U^{m}_{j-1}\quad\text{for}\;\;j\neq 0,\\
   \label{eq:DF2}
   U^{m+1}_0&=&\alpha U^{m+1}_1+(1-\alpha)U^{m+1}_{-1}.
   \end{eqnarray}
It is obvious that \eqref{eq:CI-DIF1}-\eqref{eq:DF2} is equivalent to \eqref{eq:algo-rec}, in other words $U^M_j=v_0(j)$ for any $j\in \Z$.

\vspace{0.2cm}

Let us explain briefly why we may interpret \eqref{eq:CI-DIF1}-\eqref{eq:DF2} as an explicit finite difference  scheme.
The simplest way is to examine the case $k=1$. Then $\alpha=\frac12$ and \eqref{eq:DF1}-\eqref{eq:DF2} becomes
\begin{equation}
\label{eq:DF3}
U^{m+1}_j=\frac {U^{m}_{j+1}+U^{m}_{j-1}}{2}\quad\text{for any}\;\;j\in\Z.\\
\end{equation}
Besides \eqref{eq_cauchy} becomes simply the heat equation
\begin{equation}
\label{eq:chaleur}
\begin{array}{lll}
 \partial_tu(t,x) &=&\dfrac1 2 u^{\prime\prime}(t,x) ,\quad x\in\R,\;\,t>0\\
 \\
 u(0,x)&=&u_0(x)\quad x\in\R
 \end{array}
\end{equation}
Performing an explicit finite difference scheme with space step $h$ and time step $\delta t$ for Eq.  \eqref{eq:chaleur} amounts to considering a space grid $\{x^h_j\}_{j\in\Z}$ defined by $x^h_j=jh$ for any $j\in\Z$, and to compute $\{U_j^m\}$, for $j\in\Z$, $0\leq m\leq M$, by
$$U_j^0=u_0(x^h_j)\quad j\in\Z$$
 and
$$
\frac{U^{m+1}_j-U^m_j}{\delta t}=\frac{U^m_{j+1}-2U^m_j+U^m_{j-1}}{2h^2}\quad\text{for any}\;\;j\in\Z
$$
for any $0\leq m\leq N-1$. Taking $h=n^{-1/2}$ and $\delta t=n^{-1}$ (note that this corresponds to touching the bound giving the CFL condition, see e.g. \cite{allaire}) we get \eqref{eq:CI-DIF1} and \eqref{eq:DF3}. Therefore the interpretation.
In fact it seems that the transition (pseudo)  probabilities of the random walk $S$ suggests how to take into account the transmission condition in \eqref{eq:CT-cauchy} in a finite difference scheme for \eqref{eq_cauchy}, leading to condition~\eqref{eq:DF2}.

\vspace{0.2cm}
Note that by applying the scheme \eqref{eq:CI-DIF1}-\eqref{eq:DF2} we get for any $j\in\Z$, and any $0\leq m\leq N$ an approximation $U_j^m$ of $u(\frac{m}{n},x_j=\frac{j}{\sqrt{n}})$.

For computational purposes we have to consider a PDE problem with a bounded space domain, and this is our PDE of interest \eqref{eq_evol}. 

Firstly, the domain $[-a,a]$ is discretized with a grid $\{x_j\}_{j=-N_a}^{N_a}$ with $N_a=\sqrt{n}\,a$ (we assume this quantity is an integer) and $x_j=j/\sqrt{n}$ for $-N_a\leq j\leq N_a$.

Secondly, we have to adapt \eqref{eq:CI-DIF1} and \eqref{eq:DF2} to a bounded domain (see Algorithm 2 just hereafter) and thirdly,
 we have to take into account  the Dirichlet boundary conditions  by imposing  $U^m_{-N_a}=U^m_{N_a}=0$ for  $1\leq m\leq N$. 

\vspace{0.3cm}

\fbox{\parbox{\textwidth}{
{\bf ALGORITHM 2:} Computation of $\bar{u}_{RW}^n(t,x)$.

\vspace{0.2cm}
{\bf Parameters of the method:} A time horizon $0<T<\infty$ and a discretization order $n\in\N^*$.

We set $N=nT$ and $N_a=\sqrt{n}\,a$ and assume this quantities are integers.

We set $x_j=j/\sqrt{n}$ for $-N_a\leq j\leq N_a$.

\vspace{0.3cm}
{\bf Algorithm:} 1) Set $U_j^0=u_0(x_j)$ for any $-N_a+1\leq j\leq N_a-1$.

\vspace{0.2cm}
2) For $0\leq m\leq N-1$, compute
 \begin{eqnarray*}
   U^{m+1}_j&=&\frac12 U^{m}_{j+1}+\frac 1 2U^{m}_{j-1}\quad\text{for}\;\;-N_a+1\leq j\leq N_a-1\\
   U^{m+1}_0&=&\alpha U^{m+1}_1+(1-\alpha)U^{m+1}_{-1}
       \end{eqnarray*}
   with the convention that $U^m_{-N_a}=U^m_{N_a}=0$.
   
   \vspace{0.3cm}
   
3) Return a piecewise constant function $\bar{u}_{RW}^n(t,x)$ satisfying 
$$\bar{u}_{RW}^n(\frac m n, x_j)=U^m_j,\quad\forall-N_a+1\leq j\leq N_a-1,\;\; \forall 0\leq m\leq M,$$
and
$$\bar{u}_{RW}^n(\frac m n, \pm a)=0, \quad\forall 1\leq m\leq M.$$
   }}

\vspace{0.3cm}
It should be possible to adapt the results of Proposition \ref{prop:conv-pseudo-MA} to prove convergence of the above scheme.
This would require considering the trajectories of  killed scaled pseudo asymmetric random walks, which presents difficulties we have decided not to address in the present paper.

\vspace{0.1cm}
Nevertheless, we  suspect that the function
 $\bar{u}^n_{RW}:[0,T]\times[-a,a]\to\R$ defined by Algorithm 2 should converge towards $u$ the solution of \eqref{eq_evol}:

\begin{equation}
\label{eq:conv-algo2}
\sup_{(t,x)\in[0,T]\times[-a,a]}|u(t,x)-\bar{u}^n_{RW}(t,x)|\xrightarrow[n\to\infty]{}0.
\end{equation}
This will be illustrated by the numerical experiments of Section \ref{sec:experiments}.

\subsection{Scheme inspired by the expression of the fundamental solution involving the transition function of the killed Brownian motion}
\label{ssec:num-kill-bro}

Let $u_0\in C(I;R)$. We denote~$u_0^+=u_0\1_{\R_+}$
and $u_0^-=u_0\1_{\R_-^*}$. Let $x\in (0,a)$, from Lemma \ref{lem:verif} and Eq.~\eqref{eq:trans} we have 
\begin{equation}
\label{eq:combin1}
\begin{array}{lll}
u(t,x)&=&\ds\int_{-a}^au_0(y)\check{p}(t,x,y)dy\\
\\
&=&\ds (1-\beta)\int_{-a}^a u_0^-(y)p^{(-a,a)}_W(t,x,y)dy+(1+\beta)\int_{-a}^a u_0^+(y)p^{(-a,a)}_W(t,x,y)dy\\
\\
&&\ds\hspace{0.4cm}-\beta\int_0^a u_0^+(y)p^{(0,a)}_W(t,x,y)dy\\
\\
&=&(1-\beta)u_{1}(t,x)+(1+\beta)u_{2}(t,x)-\beta u_{3,+}(t,x),\\
\end{array}
\end{equation}
where the functions $u_1$, $u_2$ and $u_{3,+}$ are respectively solution of the following heat equations:
\begin{eqnarray*} 
\left\{ \begin{array}{lcll}
\partial_t u_{1}(t,x) &=& \ds\frac{1}{2}\partial^2_{xx}u_{1}(t,x),
&\quad x \in I, t > 0,
\\[8pt]
u_{1}(0,x) &=& u_0^-(x) & \quad x \in I,
\\[8pt]
u_{1}(t,\pm a) &=& 0, & t > 0,
\end{array}\right.
\end{eqnarray*}
\begin{eqnarray*} 
\left\{ \begin{array}{lcll}
\partial_t u_{2}(t,x) &=& \ds\frac{1}{2}\partial^2_{xx}u_{2}(t,x),
&\quad x \in I, t > 0,
\\[8pt]
u_{2}(0,x) &=& u_0^+(x) & \quad x \in I,
\\[8pt]
u_{2}(t,\pm a) &=& 0, & t > 0,
\end{array}\right.
\end{eqnarray*}
and
\begin{eqnarray*} 
\left\{ \begin{array}{lcll}
\partial_t u_{3,+}(t,x) &=& \ds\frac{1}{2}\partial^2_{xx}u_{3,+}(t,x),
&\quad x \in (0,a), t > 0,
\\[8pt]
u_{3,+}(0,x) &=& u_0^+(x) & \quad x \in (0,a),
\\[8pt]
u_{3,+}(t,0)= u_{3,+}(t,a)&=& 0, & t > 0.
\end{array}\right.
\end{eqnarray*}
Indeed $p^{(-a,a)}_W(t,x,y)$ (resp. $p^{(0,a)}_W(t,x,y)$) may be viewed as the fundamental solution of the heat equation with half Laplacian on the domain $(-a,a)$ (resp. $(0,a)$), with homogeneous Dirichlet boundary conditions~(\cite{borodin}, Appendix I, Nï¿½ 6).

In the same manner, for $x<0$ we have
\begin{equation}
\label{eq:combin2}
u(t,x)=(1-\beta)u_{1}(t,x)+(1+\beta)u_{2}(t,x)+\beta u_{3,-}(t,x),
\end{equation}
with $u_1$ and $u_2$ as before and $u_{3,-}$ the solution of 
\begin{eqnarray*} 
\left\{ \begin{array}{lcll}
\partial_t u_{3,-}(t,x) &=& \ds\frac{1}{2}\partial^2_{xx}u_{3,-}(t,x),
&\quad x \in (-a,0), t > 0,
\\[8pt]
u_{3,-}(0,x) &=& u_0^-(x) & \quad x \in (-a,0),
\\[8pt]
u_{3,-}(t,-a)= u_{3,-}(t,0)&=& 0, & t > 0.
\end{array}\right.
\end{eqnarray*}
Our idea is to perform finite different schemes for $u_1$, $u_2$, $u_{3,\pm}$ and to combine them through \eqref{eq:combin1}\eqref{eq:combin2} in order to get a scheme for the approximation of $u$. We sum up the procedure in Algorithm 3 where we use implicit finite different schemes, which are known to be unconditionably stable (\cite{allaire}). In the present case they are also consistent and therefore convergent by Lax principle.

\vspace{0.3cm}

\fbox{\parbox{\textwidth}{
{\bf ALGORITHM 3:} Computation of $\bar{u}_{fund}^{h,N}(t,x)$.

\vspace{0.2cm}
{\bf Parameters of the method:} A time horizon $0<T<\infty$. 

A time dicretization order $N$ is given and we set $\delta t=T/N$.

A space step $h$ is given and we set $N_a=a/h$ (we assume this is an integer).

We set $x_j=jh$ for $-N_a\leq j\leq N_a$.

\vspace{0.3cm}
{\bf Algorithm:} 1) Set $U_{1,j}^0=u_0^-(x_j)$ for any $-N_a+1\leq j\leq N_a-1$.

Set $U_{2,j}^0=u_0^+(x_j)$ for any $-N_a+1\leq j\leq N_a-1$.

Set $U_{3+,j}^0=u_0^+(x_j)$ for any $1\leq j\leq N_a-1$.

Set $U_{3-,j}^0=u_0^-(x_j)$ for any $-N_a+1\leq j\leq -1$.

\vspace{0.2cm}
2) For $0\leq m\leq N-1$, compute the vectors $U^{m+1}_{1,}$, $U^{m+1}_{2,}$, $U^{m+1}_{3\pm,}$ by applying the implicit finite difference schemes
 \begin{eqnarray*}
   \frac{U^{m+1}_{1,j}-U^{m}_{1,j}}{\delta t}&=&\frac{U^{m+1}_{1,j+1}-2U^{m+1}_{1,j}+U^{m+1}_{1,j-1}}{2h^2}\quad\text{for}\;\;-N_a+1\leq j\leq N_a-1\\
   \frac{U^{m+1}_{2,j}-U^{m}_{2,j}}{\delta t}&=&\frac{U^{m+1}_{2,j+1}-2U^{m+1}_{2,j}+U^{m+1}_{2,j-1}}{2h^2}\quad\text{for}\;\;-N_a+1\leq j\leq N_a-1\\
    \frac{U^{m+1}_{3+,j}-U^{m}_{3+,j}}{\delta t}&=&\frac{U^{m+1}_{3+,j+1}-2U^{m+1}_{3+,j}+U^{m+1}_{3+,j-1}}{2h^2}\quad\text{for}\;\;1\leq j\leq N_a-1\\
    \frac{U^{m+1}_{3-,j}-U^{m}_{3-,j}}{\delta t}&=&\frac{U^{m+1}_{3-,j+1}-2U^{m+1}_{3-,j}+U^{m+1}_{3-,j-1}}{2h^2}\quad\text{for}\;\;-N_a+1\leq j\leq -1\\
   \end{eqnarray*}
   with the conventions that $U^{m+1}_{1,\pm N_a}=U^{m+1}_{2,\pm N_a}=U^{m+1}_{3+, N_a}=U^{m+1}_{3+,0}=U^{m+1}_{3-, -N_a}=U^{m+1}_{3-,0}=0$.
   
   \vspace{0.3cm}
   
3) Return a piecewise constant function $\bar{u}_{fund}^{h,N}(t,x)$ satisfying 
$$\bar{u}_{fund}^{h,N}(m\,\delta t, x_j)=(1-\beta)U^m_{1,j}+(1+\beta)U^m_{2,j}-\beta U^m_{3+,j},\quad\forall 1\leq j\leq N_a-1,\;\; \forall 0\leq m\leq M,$$
$$\bar{u}_{fund}^{h,N}(m\,\delta t, x_j)=(1-\beta)U^m_{1,j}+(1+\beta)U^m_{2,j}+\beta U^m_{3-,j},\quad\forall -N_a+1\leq j\leq -1,\;\; \forall 0\leq m\leq M,$$
$$\bar{u}_{fund}^{h,N}(m\,\delta t, 0)=(1-\beta)U^m_{1,0}+(1+\beta)U^m_{2,0},\quad \forall 0\leq m\leq M,$$
and
$$\bar{u}_{fund}^{h,N}(m\,\delta t, \pm a)=0, \quad\forall 1\leq m\leq M.$$
   }}
   
   \vspace{0.3cm}

From the convergence of the finite difference schemes we  immediately get the following convergence result.

\begin{proposition}
Let us consider $u$ the solution of \eqref{eq_evol}. Let $0<T<0$ and let us consider for any~$N\in\N^*$ and any $h\in(0,a)$ the function $\bar{u}^{h,N}_{fund}:[0,T]\times[-a,a]\to\R$ defined by Algorithm 3. We have
$$
\sup_{(t,x)\in[0,T]\times[-a,a]}|u(t,x)-\bar{u}^{h,N}_{fund}(t,x)|\xrightarrow[h\downarrow 0,\,N\to\infty]{}0.
$$
\end{proposition}

\begin{remark}
In fact we could infer several other numerical schemes from Eq. \eqref{eq:trans}. For example, consider $x>0$ and $t>0$. 
Eq. \eqref{eq:combin1} also implies that
$$
u(t,x)=(1-\beta)\E^x[u_0^-(W_t);\,t<\tau_{(-a,a)}]+(1+\beta)\E^x[u_0^+(W_t);\,t<\tau_{(-a,a)}] 
- \beta \E^x[u_0^+(W_t);\,t<\tau_{(0,a)}].
$$
So we could consider approaching each of the above expectations by Monte Carlo sums involving samples of independent Brownian motions and the corresponding stopping times $\tau_{(-a,a)}$ or $\tau_{(0,a)}$ (see e.g. \cite{elkaroui-gobet}).

However, as the space dimension  is one, we know that this Monte Carlo method would be slower than the finite difference approach described in Algorithm 3. Nevertheless such an approach could be interesting if we would address the problem in a space of higher dimension.
\end{remark}

\section{Numerical experiments}
\label{sec:experiments}

{\bf Example 1.} We take $I=(-1,1)$ (i.e. $a=1$) and $k=-0.5$. We choose the following initial condition
$$
u_0(x)=\frac{10 x^3 - 3 x^2 - 9x + 4}{2},\quad\forall x\in I.
$$
Indeed in order to use Algorithm 1 we have to compute the $a_n$'s and $b_n$'s, through Eq. \eqref{eq:def-coeff-ab}. By the polynomial nature of the initial condition $u_0$ these coefficients will be made explicit, providing thus a benchmark for the finite difference scheme inspired algorithms (Algo. 2 and 3).

Remember that the quantites $\int_I A(x) |f_n(x)|^2 \,dx$ and $\int_I A(x) |g_n(x)|^2 \,dx$ in Eq. \eqref{eq:def-coeff-ab} are given by Eq. \eqref{eq_coeff}.
Besides one can compute
$$
\begin{array}{l}
  \ds\int_{-1}^1 A(x) f_n(x) u_0(x) \, dx \\
\\ 
=
\ds{\frac {k}{ \left( 2\,n-1 \right) ^{4}{\pi}^{4}} \left( -528\, \left( 
n-1/2 \right) \pi\, \left( -1 \right) ^{n}-480-72\, \left( n-1/2
 \right) ^{2}{\pi}^{2} \right) }
 \\
\hspace{0.4cm}\ds+{\frac {1}{ \left( 2\,n-1 \right) ^{4
}{\pi}^{4}} \left( -528\,\pi\, \left( 1/33\, \left( n-1/2 \right) ^{2}
{\pi}^{2}-{\frac{9}{11}} \right)  \left( n-1/2 \right)  \left( -1
 \right) ^{n}+480+72\, \left( n-1/2 \right) ^{2}{\pi}^{2} \right) }\\
\end{array}
$$
and
$$
 \int_{-1}^1 A(x) g_n(x) u_0(x) \, dx 
\;=\;
-{\frac {k \left( -1 \right) ^{n} \left( {\pi}^{2}{n}^{2}-60 \right) 
}{{\pi}^{3}{n}^{3}}}.
$$
Recalling that here $\lambda_n=\frac{(2n-1)\pi}{2}$ and $\mu_n=n\pi$, $n\geq 1$, we have everything at hand to perform Algorithm~1.

\vspace{0.1cm}
For performing Algorithms 2 and 3 no previous computation is needed.

\vspace{0.2cm}
Figure \ref{fig1} represents the graphs of  $\bar{u}_{spec}^{200}(T,\cdot)$, $\bar{u}_{RW}^{2.5\times 10^5}(T,\cdot)$ and 
$\bar{u}_{fund}^{2\times 10^{-3}\,,\,500}(T,\cdot)$ at $T=0.4$ with the following choices of parameters:
 $N=200$ for Algorithm 1, 
 $n=2.5\times 10^5$ for Algorithm 2,
 a space step $h=2\times 10^{-3}$ and a time discretization order $N=500$ for Algorithm 3.

We see a very good concordance between the three methods (which actually can be observed for coarser discretizations). 

In particular we can numerically check the convergence of the Algorithm 2 announced in Eq. \eqref{eq:conv-algo2}. To that aim we 
consider $\bar{u}_{spec}^{200}(T,\cdot)$ as the reference solution and
report in Table \ref{tab:1} the value of
$$\sup_{x\in[-a,a]}\big|\bar{u}_{spec}^{200}(T,x)-\bar{u}_{RW}^{n}(T,x)\big|$$
for increasing values of $n$, up to $n=2.5\times 10^5$. Convergence is indeed observed.
\vspace{1cm}

\begin{table}[h!]
 \begin{center}
\begin{tabular}{cc}
\hline
&\\
  n  &   $\sup_{x\in[-a,a]}\big|\bar{u}_{spec}^{200}(T,x)-\bar{u}_{RW}^{n}(T,x)\big|$  \\
  \\
 \hline
$100$ & $4.56\times 10^{-2}$ \\
$625$ & $9.53\times 10^{-3}$\\
$10^4$ & $4.4\times 10^{-4}$\\
$2.5\times 10^5$ & $1.7\times 10^{-5}$\\
\hline
\end{tabular}
\caption{Approximation error $\sup_{x\in[-a,a]}\big|\bar{u}_{spec}^{200}(T,x)-\bar{u}_{RW}^{n}(T,x)\big|$ in function of $n$ (the function
$\bar{u}_{spec}^{200}(T,\cdot)$ is considered as the benchmark).}
\label{tab:1} 
\end{center}
\end{table}

\begin{figure}
\begin{center}
\includegraphics[width=13cm]{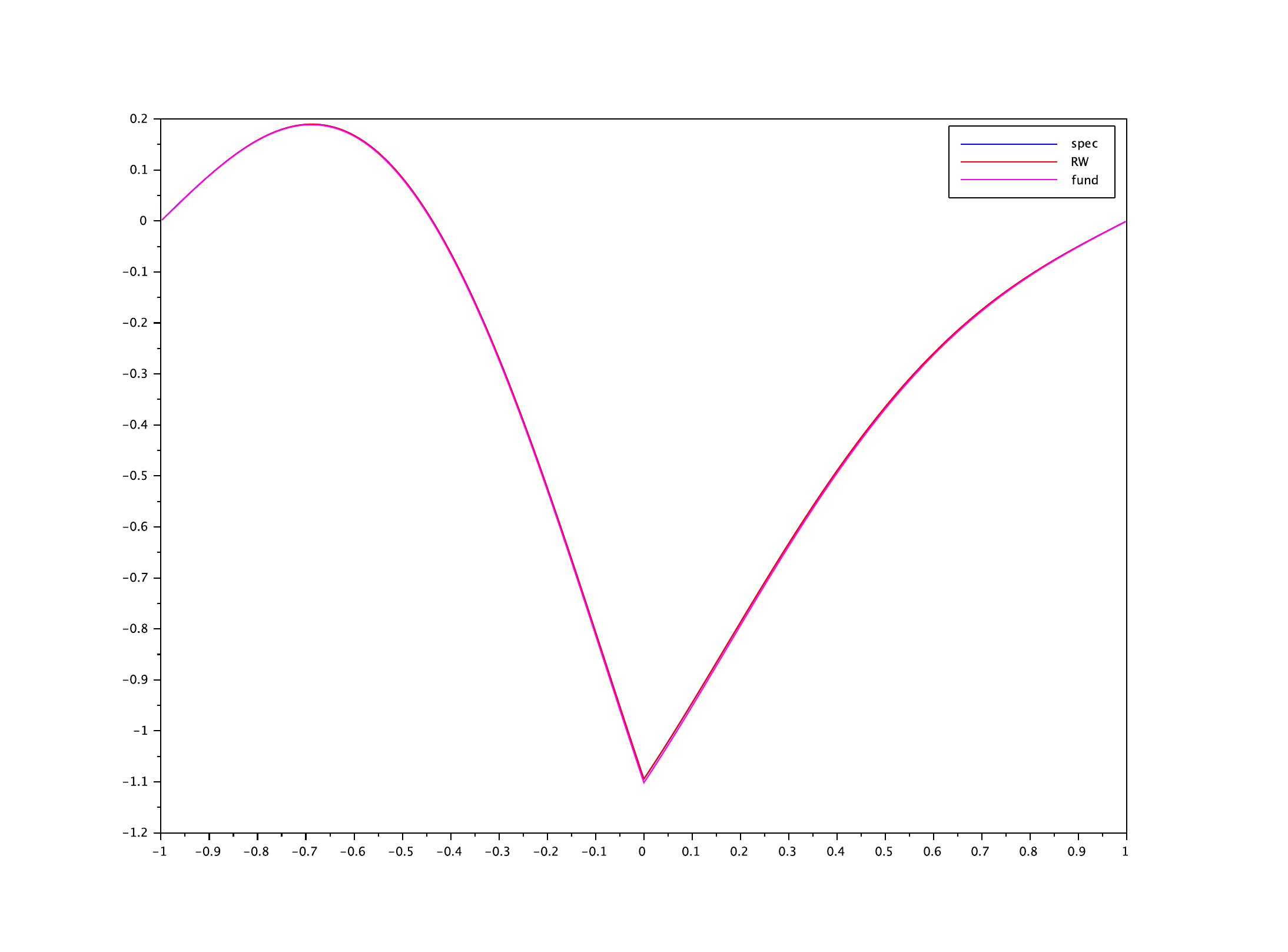}
\caption{ Plot of an approximation of the function $u(T=0.4,\cdot)$, by $\bar{u}_{spec}^{200}(T,\cdot)$, $\bar{u}_{RW}^{2.5\times 10^5}(T,\cdot)$ and 
$\bar{u}_{fund}^{2\times 10^{-3}\,,\,500}(T,\cdot)$, for the initial condition 
$u_0(x)=(10 x^3 - 3 x^2 - 9x + 4)/2$.}
\label{fig1}
\end{center}
\begin{center}
\includegraphics[width=13cm]{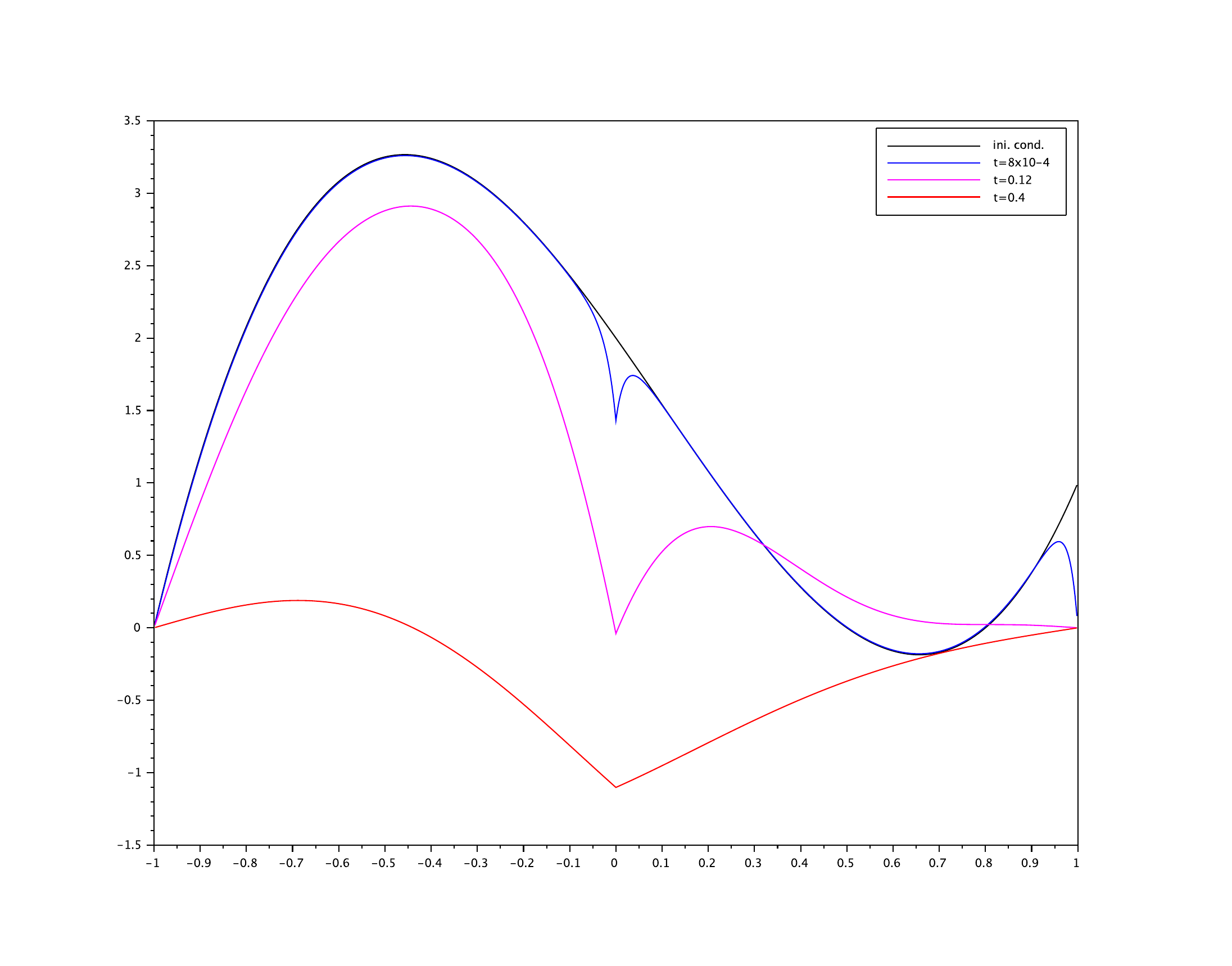}
\caption{ Plot of the initial condition $u_0$ and of an approximation of the function $u(t,\cdot)$, by $\bar{u}_{fund}^{2\times 10^{-3}\,,\,500}(t,\cdot)$, at times $t=8\times 10^{-4}$, $t=0.12$ and $t=0.4$, for the initial condition 
$u_0(x)=(10 x^3 - 3 x^2 - 9x + 4)/2$.}
\label{fig2}
\end{center}
\end{figure}

\vspace{0.3cm}

On Figure \ref{fig2} we check another interesting phenomenon. We plot 
 the initial condition $u_0$ and an approximation of $u(t,\cdot)$ by $\bar{u}_{fund}^{h,N}(t,\cdot)$ (we keep the same parameters $h=2\times 10^{-3}$ and $N=500$ as previously)
 at times $t=8\times 10^{-4}$, $t=0.12$ and~$t=0.4$. Note that $t=8\times 10^{-4}=T/N$ corresponds to the first time step in the finite difference scheme.

Observe that the transmission condition destroys the regularity of the initial datum. Also observe that the slope of the graph of $u(t,\cdot)$ is of negative sign at $0-$, but twice bigger in absolute value than the positive slope at $0+$. This is what we expect as 
$k=-0.5$.

\vspace{0.7cm}
{\bf Example 2.} In this second example we keep $I=(-1,1)$  and $k=-0.5$, but choose the following initial condition
\begin{equation}
\label{eq:def-CI-creneau}
u_0(x)=\1_{x<0}-0.5,\quad\forall x\in I.
\end{equation}
Indeed we want to test the robustness of our numerical schemes to a non smooth initial condition, especially if this initial condition presents a discontinuity at the interface point $x=0$.

Of course for Algorithm 1 the initial condition can be taken in $L^2(I)$ so one knows that the algorithm converges if we take $u_0$ defined by 
\eqref{eq:def-CI-creneau}.

The analysis of probabilistic schemes such as Algorithm 2 usually relies on arguments of convergence in pseudo law, which are usually valid only for smooth functions.  We suspect that the smoothing properties of the operator $\frac{1}{2A}\nabla\cdot(A\nabla\,)$ could be used to prove convergence of Algorithm 2,
 even for non smooth initial conditions.

\vspace{0.2cm}
In order to use Algorithm 1 we have to compute the $a_n$'s and $b_n$'s again, the $\mu_n$'s and $\lambda_n$'s remaining unchanged. Easy computations show that
$$
a_n=-\frac{(-1)^n2(k-1)}{(2n-1)\pi(k+1)}$$
and
$$
b_n=\frac{4}{n\pi(k+1)}\1_{n\text{ is odd}},
$$
for $n\geq 1$.

We take again an order of truncation $N=200$ for Algorithm 1.

Again for the Algorithm 2 we take a discretization order $n=2.5\times 10^5$, and
for the Algorithm 3 we take a space step $h=2\times 10^{-3}$ and a time discretization order $N=500$.

\vspace{0.1cm}
 Figure \ref{fig3} depicts the graphs of $\bar{u}_{spec}^{200}(T,\cdot)$, $\bar{u}_{RW}^{2.5\times 10^5}(T,\cdot)$ and 
$\bar{u}_{fund}^{2\times 10^{-3}\,,\,500}(T,\cdot)$ for $T=0.4$ and again shows very good agreement between the results of the three schemes.

\vspace{0.3cm}

On Figure \ref{fig4} we  plot 
 the initial condition $u_0$ and an approximation of $u(t,\cdot)$ by $\bar{u}_{fund}^{h,N}(t,\cdot)$ (we keep the same parameters $h=2\times 10^{-3}$ and $N=500$ as previously)
 at times $t=0.012$, $t=0.12$ and~$t=0.4$. 
In particular, the plot illustrates the fact that the solution $u(t,x)$ does not satisfy a maximum principle, one of the main arguments classically used in the numerical analysis of deterministic finite difference schemes for parabolic equations.

\begin{figure}
\begin{center}
\includegraphics[width=12.5cm]{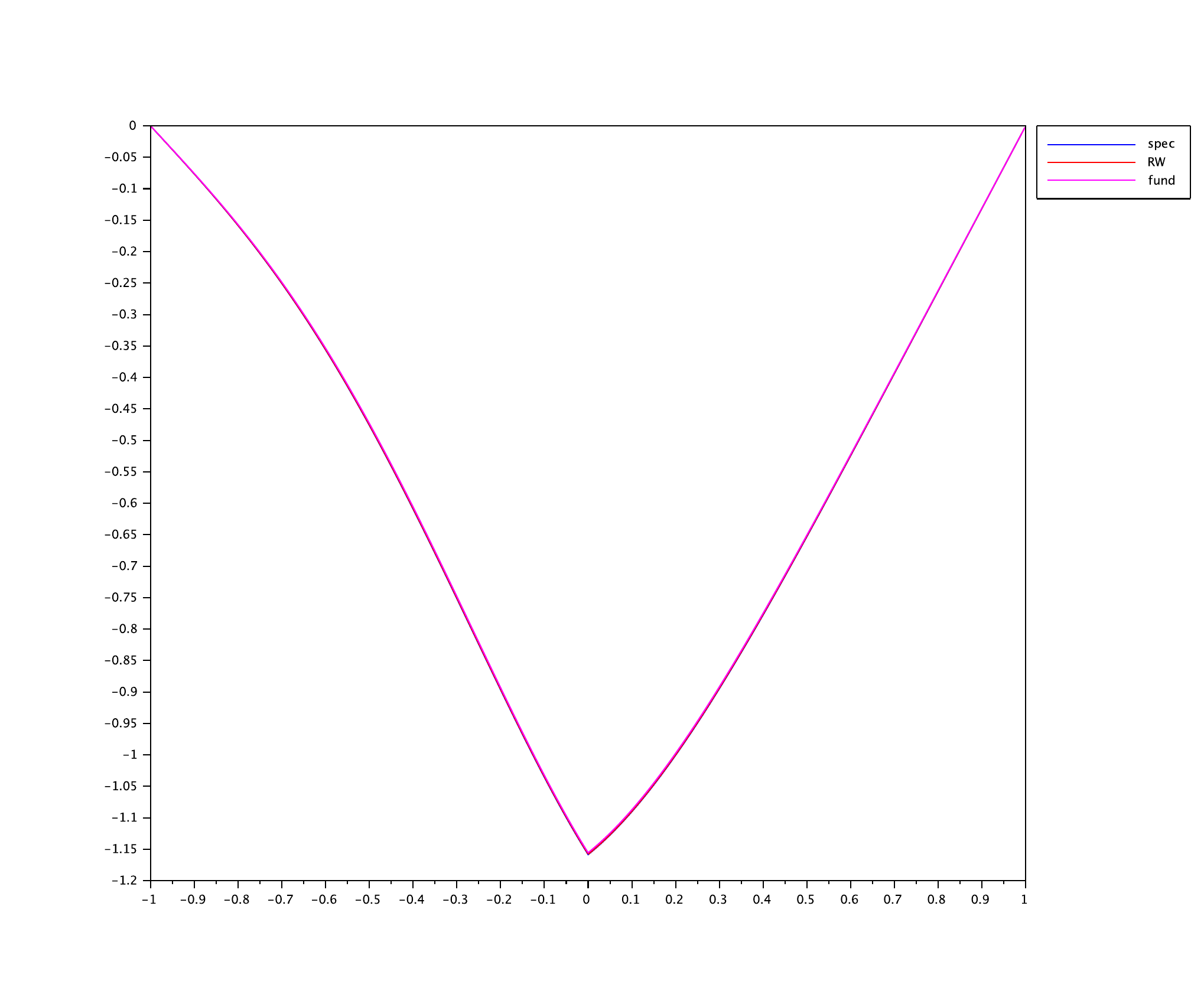}
\caption{ Plot of an approximation of the function $u(T=0.4,\cdot)$, by $\bar{u}_{spec}^{200}(T,\cdot)$, $\bar{u}_{RW}^{2.5\times 10^5}(T,\cdot)$ and 
$\bar{u}_{fund}^{2\times 10^{-3}\,,\,500}(T,\cdot)$, for the initial condition 
$u_0(x)=\1_{x<0}-0.5$.}
\label{fig3}
\end{center}
\begin{center}
\includegraphics[width=12.5cm]{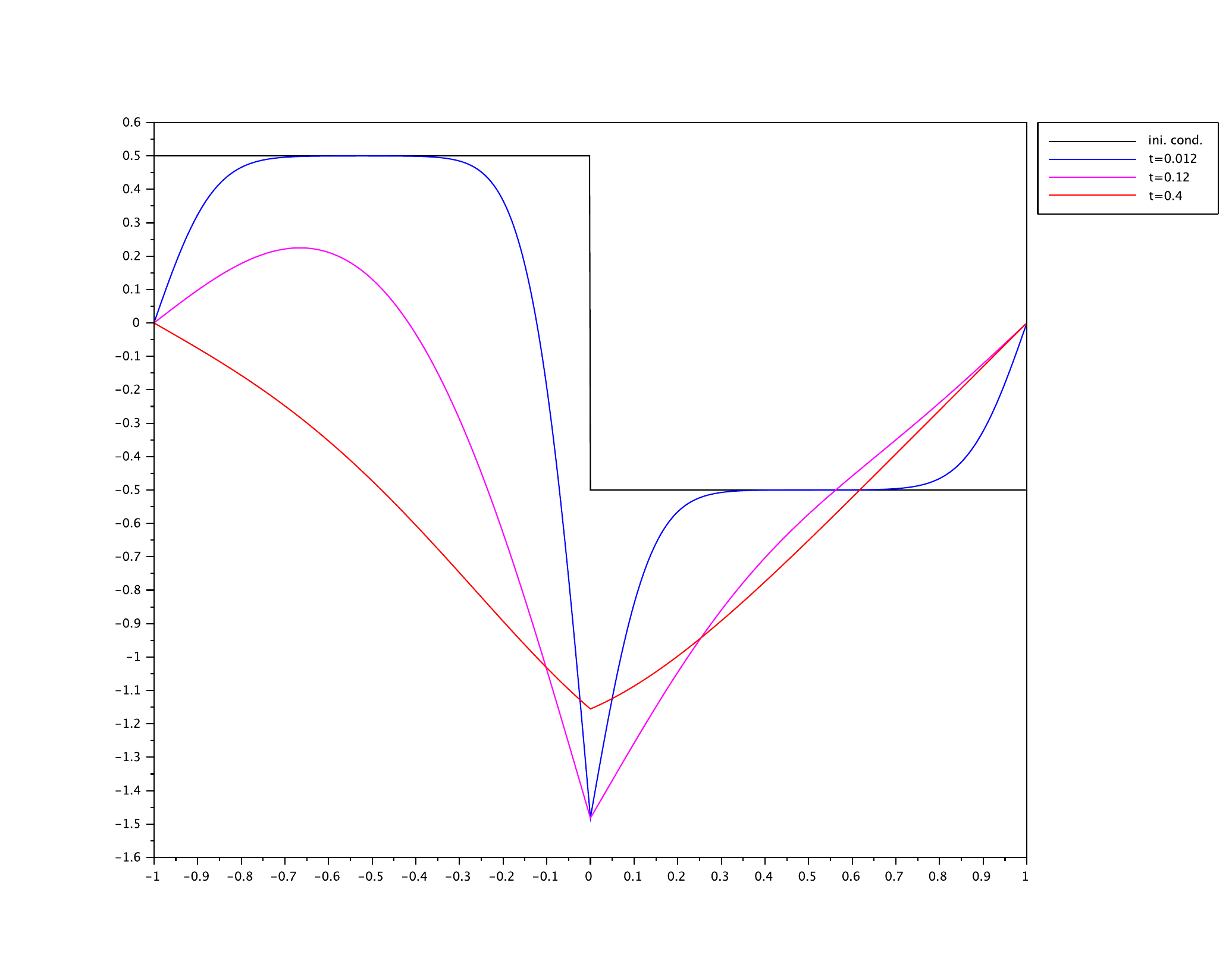}
\caption{ Plot of the initial condition $u_0$ and of an approximation of the function $u(t,\cdot)$, by $\bar{u}_{fund}^{2\times 10^{-3}\,,\,500}(t,\cdot)$, at times $t=0.012$, $t=0.12$ and $t=0.4$, for the initial condition 
$u_0(x)=\1_{x<0}-0.5$.}
\label{fig4}
\end{center}
\end{figure}

\vspace{0.5cm}

{\bf Fundings}

This research has been supported by the IRS projects "FKSC" and "Spectral properties of graphs with negative index materials" (for E. Bonnetier and P. \'Etoré) and by  the Bézout Labex, funded by ANR, reference ANR-10-LABX-58 (for Miguel Martinez).

\vspace{0.1cm}

{\bf Data availability and conflict of interest statements}

No datasets were generated or analysed during the current study.

The authors declare no competing interests.

	 \bibliographystyle{amsplain}
\bibliography{BIB_sign_changing.bib}

\end{document}